\pdfoutput=1
\documentclass[acmsmall,nonacm]{acmart}

\usepackage{booktabs} 
\usepackage{courier}
\usepackage{graphicx}
\usepackage{graphics}
\usepackage{url}
\usepackage{multirow}
\usepackage{xcolor}
\usepackage{subfig}
\usepackage{balance}
\usepackage{amsmath}
\usepackage{bbm}
\usepackage{algorithmicx}
\usepackage{algorithm}
\usepackage{algpseudocode}
\usepackage{pgfplots}
\usepackage{bm}
\usepackage{cleveref}
\usepackage{hyperref}
\usepackage{pdfpages}
\def\NoNumber#1{{\def\alglinenumber##1{}\State #1}\addtocounter{ALG@line}{-1}}

\crefformat{section}{\S#2#1#3} 
\crefformat{subsection}{\S#2#1#3}
\crefformat{subsubsection}{\S#2#1#3}
\newenvironment{customlegend}[1][]{%
	\begingroup
	\csname pgfplots@init@cleared@structures\endcsname
	\pgfplotsset{#1}%
}{%
\csname pgfplots@createlegend\endcsname
\endgroup
}%
\def\addlegendimage{\csname pgfplots@addlegendimage\endcsname}
\newcommand{\EFone}{\mathrm{EF1}}

\newcommand{\MMS}{\mathrm{MMS}}
\newcommand{\ALG}{\textsc{Alg}}

\DeclareMathOperator*{\argmax}{arg\,max}

\newcommand{\Mod}[1]{\ (\mathrm{mod}\ #1)}
\newtheorem{theorem}{Theorem}[section]
\newtheorem{lemma}{Lemma}[theorem]
\newtheorem{corollary}{Corollary}[lemma]
\newtheorem{proposition}[theorem]{Proposition}

\usepackage{enumitem}
\newcommand{\subscript}[2]{$#1 _ #2$}

\AtBeginDocument{%
  \providecommand\BibTeX{{%
    \normalfont B\kern-0.5em{\scshape i\kern-0.25em b}\kern-0.8em\TeX}}}

\setcopyright{none}
\begin{document}
\title{Towards Fair Recommendation in Two-Sided Platforms}

\author{Arpita Biswas}
\authornotemark[1]
\affiliation{
\institution{Harvard University}
\country{USA}
}

\author{Gourab K Patro}
\authornote{Both authors contributed equally to this work.}
\affiliation{
	\institution{Indian Institute of Technology Kharagpur}
	\country{India}
}
\affiliation{
	\institution{L3S Research Center}
	\country{Germany}
}

\author{Niloy Ganguly}
\affiliation{
\institution{Indian Institute of Technology Kharagpur}
\country{India}
}
\affiliation{
\institution{L3S Research Center}
\country{Germany}
}

\author{Krishna P. Gummadi}
\affiliation{
\institution{Max Planck Institute for Software Systems, Germany}
\country{Germany}
}

\author{Abhijnan Chakraborty}
\affiliation{
\institution{Indian Institute of Technology Delhi}
\country{India}
}

\renewcommand{\shortauthors}{}

\begin{abstract}
Many online platforms today (such as Amazon, Netflix, Spotify, LinkedIn or AirBnB) can be thought of as 
two-sided markets with producers and customers of goods and services. Traditionally, recommendation services in these platforms have focused on maximizing customer satisfaction by tailoring the results according to the personalized preferences of individual customers. However, our investigation reinforces the fact that such customer-centric design of these services may lead to unfair distribution of exposure to the producers, which may adversely impact their well-being. On the other hand, a pure producer-centric design might become unfair to the customers. As more and more people are depending on such platforms to earn a living, it is important to ensure fairness to both producers and customers. In this work, by mapping a fair personalized recommendation problem to a constrained version of the problem of fairly allocating indivisible goods, we propose to provide fairness guarantees for both sides. 
Formally, our proposed {\em FairRec} algorithm guarantees Maxi-Min Share ($\alpha$-MMS) of exposure for the producers, 
and Envy-Free up to One Item ($\EFone$) fairness for the customers. Extensive evaluations over multiple real-world datasets show the effectiveness of {\em FairRec} in ensuring two-sided fairness while incurring a marginal loss in overall recommendation quality. Finally, we present a modification of {\it FairRec} (named as {\it FairRecPlus}) that at the cost of additional computation time, improves the recommendation performance for the customers, while maintaining the same fairness guarantees.
\end{abstract}



\clubpenalty = 10000
\widowpenalty = 10000

\setlength{\belowdisplayskip}{1.1pt} 
\setlength{\belowdisplayshortskip}{1.1pt}
\setlength{\abovedisplayskip}{1.1pt}
\setlength{\abovedisplayshortskip}{1.1pt}

\maketitle

\section{introduction}\label{introduction}
Popular online platforms such as Netflix, Amazon, Yelp, Spotify, Google Local provide recommendation services to help their customers browse through the enormous product spaces.
These recommendation services often play a huge role in controling the interaction between the two stakeholders, namely (i) {\bf producers} of goods and services (e.g., movies on Netflix, products on Amazon, restaurants on Yelp, artists on Spotify) and (ii) {\bf customers} who consume them.
Maximization of customer satisfaction has been the traditional focus on these platforms which is often achieved by tailoring the recommendations according to the personalized preferences of individual customers, largely ignoring the interest of the producers.
Several recent studies have shown how such customer-centric designs may undermine the well-being of the producers \cite{abdollahpouri2019multi,burke2017multisided,edelman2017racial,graham2017digital,hannak2017bias}.
As more and more people are depending on two-sided platforms to earn a living,
recently platforms have also started showing interest in creating fair marketplaces for all the stakeholders due to multiple reasons:
{\bf (i) legal obligation} (e.g., labor bill for the welfare of drivers on Uber and Lyft~\cite{nyt2019uber}, fair marketplace laws for e-commerce~\cite{fair_marketplace}),
{\bf (ii) social responsibility or voluntary commitment} (e.g., equality of opportunity to all gender groups in LinkedIn~\cite{geyik2019fairranking}, commitment of non-discrimination to hosts and guests by AirBnb~\cite{airbnb_commitment}),
{\bf (iii) business requirement/model} (e.g., minimum business guarantee by AirBnb to attract hosts~\cite{airbnb_min_guarantee}).

In this paper, our focus is on the fairness of personalized recommendation services deployed on the two-sided platforms.
Traditionally, platforms employ various state-of-the-art data-driven methods (e.g., neighborhood-based methods~\cite{ning2015comprehensive}, latent factorization methods~\cite{liang2016factorization,koren2009matrix}, etc.) to estimate the relevance scores of every product-customer pairs, and then recommend $k$ most relevant products to the corresponding customers.
While such top-$k$ recommendations achieve high customer utility, our investigation on real-world datasets reinforces the presence of popularity bias~\cite{jannach2015recommenders} that is, they can create a huge disparity in the exposure of the producers (detailed in \cref{observations}), which is unfair for the producers, and may also hurt the platforms in the long term.

In these platforms, \textit{exposure often determines the economic opportunities (revenues) for the producers} who depend on it for their livelihood.
For instance, high exposure on {\em Google Maps} can increase the footfall in a local business, thereby increasing their revenue. 
High exposure on {\em YouTube, Spotify or Last.fm} can increase the traffic to a content producer's channel, and hence help them earn better platform-royalties or advertisement revenues.
On the other hand, if only a few producers get most of the exposure, then the other producers would struggle on the platform, which could force them to either quit or switch to other platforms~\cite{leaving_yelp,leaving_amz,leaving_uber}. 
This, in turn, may limit the choices for the customers, degrading the overall experience on the platform.
Thus, it is important to reduce exposure inequalities. 
However, extremely producer-centric ways of reducing inequality (e.g., recommending the $k$ least exposed producers to the customers) may result in loss and disparity in customer utilities (\cref{observations}), making it inefficient as well as unfair to the customers.
%

To counter such unfairness for both producers and customers, we propose to tackle the challenging task of ensuring two-sided fairness while providing personalized recommendations. 
Specifically, we propose to ensure a minimum exposure guarantee for every producer such that no producer starves for exposure.
Since the exposure guarantee on the producer side could incur losses on the customer side (i.e., reduction in customer utilities), 
we propose that the loss in utility should be fairly distributed among the customers.
Motivated by a vast literature in social choice theory, we map this problem to the problem of fairly allocating indivisible goods (\cref{mapping}). 
In an allocation problem, there is often a predefined set of items and a set of agents with their valuations (how much an agent values an item), and the task is to allocate the items among the agents. 
To map our recommendation problem to an allocation problem, we assume the set of customers as the set of agents. 
Now we can strategically fix the set of items as the one which contains as many copies of each product (or producer) as the chosen exposure guarantee, and then allocate them among the 
customers. 
%
%
If we have an algorithm that does this task, then the strategic setting of the item-set and their allocation guarantee 
can, in turn, ensure minimum exposure for the producers.
Besides, the algorithm's fairness guarantee for the agents during the allocation would also be able to guarantee customer fairness.
%
Thus, the original recommendation problem becomes an interesting (constrained) extension to the existing fair allocation problem---find an allocation that guarantees minimum exposure (upper bounded by maximin share of exposure or MMS) for the producers, and envy-free up to one item (EF1)~\cite{budish2011combinatorial} 
for the customers\footnote{The $\MMS$ guarantee ensures that each agent receives a value which is at least their \emph{maximin share} threshold, defined in \cref{eq:mms}; whereas, $\EFone$ ensures that every agent values her allocation at least as much as any other agent's allocation after (hypothetically) removing the most valuable item from the other agent's allocated bundle.}.
We propose an algorithm {\it FairRec} (\cref{algorithm}) which solves 
this problem and gives guarantees on both producer and customer side (proofs in \cref{sec:theorems}).
Extensive evaluations over multiple real-world datasets show the effectiveness of FairRec in ensuring two-sided fairness while incurring a marginal loss in recommendation quality.

In summary, we make the following contributions in this paper. 
\begin{itemize}[topsep=0pt,itemsep=-0.5ex,partopsep=1ex,parsep=1ex]
	\item We consider a two-sided fair recommendation problem that not only relates to social or judicial precepts but also to the long-term sustainability of two-sided platforms (\cref{motivation}).
	\item We design an algorithm, {\it FairRec} (\cref{algorithm}), exhibiting the desired two-sided fairness by mapping the fair recommendation problem to a fair allocation problem (\cref{mapping}). Moreover, it is agnostic to the specifics of the data-driven model (that estimates the product-customer relevance scores) which makes it scalable and easy to adapt.
	\item In addition to the theoretical guarantees (\cref{sec:theorems}), 
	extensive experimentation and evaluation over multiple real-world datasets deliver strong empirical evidence 
	on the effectiveness of our proposal~(\cref{experiments}).
	\item Finally, we also present a modified version of FairRec (named as FairRecPlus) that uses an envy-cycle elimination and swapping technique to improve the performance on customer-side metrics, while maintaining the same two-sided fairness guarantees~(\cref{sec:modified_fairrec}). 
\end{itemize}
\section{Background and Related Work}\label{related}
We briefly survey related works in two directions: (i) fairness in multi-stakeholder platforms, and (ii) fair allocation of goods.

\subsection{Fairness in Two-Sided Platforms}
With the increasing popularity of multi-sided platforms, recently researchers have looked into the issues of unfairness and biases in such platforms.
For example, \citet{edelman2017racial} investigated the possibility of racial bias in guest acceptance by Airbnb hosts, 
~\citet{lambrecht2019algorithmic} studied gender-based discrimination in career ads,~\citet{chakraborty2019equality} proposed to ensure fair representation in crowdsourced recommendations. While these works deal with {\it group fairness}, ~\citet{serbos2017fairness} proposed an {\it envy-free} tour package recommendations on travel booking sites, ensuring {\it individual fairness for customers}. 

On producer fairness, \citet{hannak2017bias} studied racial and gender bias in freelance marketplaces, and \citet{dash2021umpire} investigated favoritism towards certain producers on e-commerce marketplace.
In a social experiment, \citet{Salganik854} found that the existing popular producers often acquire most of the visibility while new but good ones starve for visibility. \citet{banerjee2020analyzing} also found popularity bias in location based recommendations. \citet{kamishima2014correcting} and \citet{abdollahpouri2017controlling} proposed methods to reduce such popularity bias among producers.
\citet{surer2018multistakeholder} proposed to maximize total customer utility while ensuring some exposure for the producers.
While these works have proposed to ensure some forms of producer-side fairness, they have not looked into the resulting unfairness on the customer-side, the trade-off between producer and customer fairness, and the cost of achieving one over the other. 

Few past works have discussed fairness for both producers and customers.
~\citet{abdollahpouri2019multi} and \citet{burke2017multisided} categorized different types of multi-stakeholder platforms and their desired group fairness properties, \citet{chakraborty2017fair} and \citet{suhr2019two} presented mechanisms for two-sided fairness in matching problems while \citet{patro2020incremental} addressed fairness issues arising due to frequent updates of platforms.
In contrast, our paper addresses individual fairness for both producers and customers, which also answers the question of the long-term sustainability of two-sided platforms.

There exists another line of work on fairness in ranking and recommendations, namely \citet{singh2018fairness,biega2018equity}, which propose to ensure the producers with expected exposures in proportion to their corresponding relevance scores in order to maintain individual fairness for producers in gig-economy platforms. 
One of the limitations of these works is that they assume the availability of true relevance scores of producers.
However, these relevance scores are often estimated and usually contain noise, which is also highlighted in \citet{raj2020comparing}.
The noise is not the only issue here; if the estimated relevance scores themselves exhibit popularity bias, then ensuring exposure in proportion to these relevance scores could cause the same inequalities in producer exposures and can do very little towards producer fairness.
Thus in this work, we try to isolate the considerations of exposure and relevance. 
Formally, we use the exposure of a producer as the measure of its utility, the relevance of recommended items as the utility of a customer (more details in \cref{notions}), and finally we define fairness for both sides using the (in)equality in their individual utilities.

\subsection{Fair Allocation of Goods}
The problem of fair allocation (popularly known as the cake-cutting problem) has been studied extensively in the area of computational social choice theory. The classical notions of fairness for this problem are envy-freeness (EF)~\cite{foley1967resource,varian1974equity} and proportional fair share (PFS)~\cite{steinhaus1948problem}. 
Recent literature on practical applications of fair allocation~\cite{brandt2016handbook,endriss2017trends} has focused on the problem of allocating \textit{indivisible} goods 
in budgeted course allocation~\cite{budish2011combinatorial}, balanced graph partition~\cite{fair-graph}, or allocation of cardinality constrained group of resources~\cite{biswas2018fair}. In such instances, no feasible allocation may satisfy EF or PFS fairness guarantees. Thus, the notable work of Budish~\cite{budish2011combinatorial} defined analogous fairness notions which are appropriate for \textit{indivisible} goods---namely, envy-freeness up to one good ($\EFone$) and maximin share guarantee ($\MMS$). 

The relevance of $\EFone$ is substantiated by the fact that it is guaranteed to always exist under general monotone valuations and, in fact, such allocations can be obtained in polynomial time~\cite{lipton-envy-graph}. When the valuations are additive, Caragiannis et al~\cite{caragiannis2016unreasonable} show that a simple greedy round robin algorithm is enough to ensure $\EFone$. 

$\MMS$ fairness is another solution concept that has been extensively studied in the fair allocation space. In particular, Bouveret et al.~\cite{bouveret2014characterizing} showed that an $\MMS$ allocation exists when the agents' valuations are additive and binary (valuations are $0$ or $1$). However, Procaccia and Wang~\cite{procaccia2014fair} and Kurokawa et al.~\cite{kurokawa2016can} provided intricate counterexamples to refute the universal existence of $\MMS$ allocations, under additive and non-binary valuations. This motivated the study of approximate maximin share allocations, $\alpha$-$\MMS$, where each agent obtains a bundle of value at least $\alpha \in (0,1)$ times her maximin share. The existence of $2/3$-$\MMS$ and accompanying algorithms were developed in a sequence of results~\cite{procaccia2014fair,amanatidis2015approximation,barman2017approximation}.  Later, Ghodsi et al.~\cite{ghodsi2017fair} improved the result by providing an efficient algorithm that obtains $3/4$-$\MMS$ allocations. 

The vast majority of work in the fair allocation space has solely focused on the unconstrained version of the problem; exceptions include the work of Biswas et al.~\cite{biswas2018fair,biswas2019matroid} and Gourv{\`{e}}s et al.~\cite{gourves2017approximate,gourves2014near}. Biswas et al.~\cite{biswas2018fair} provide algorithms for computing $\EFone$ and $1/3$-$\MMS$ for the allocation problem where items are categorized into groups, and an upper bound restricts the number of items that can be allocated to each agent from each category. This is slightly different from the problem we consider (detailed in~\cref{mapping2}). A general version of the category-wise upper bound constraint, namely laminar matroid constraint, is studied by Biswas et al.~\cite{biswas2019matroid} and the existence of $\EFone$ is proved for identical valuations. A different problem is considered by
Gourv{\`{e}}s et al.~ \cite{gourves2017approximate,gourves2014near} where the goal is to find $\MMS$ fair allocation that \emph{union} of all the allocated goods is an independent set of a given matroid. Although these papers study fair allocation under several combinatorial constraints, they do not directly apply to the problem we consider. 

Moreover, all the above mentioned papers consider fairness among the agents but not among the items. In this work, we consider fairness across the agents as well as the items. We map the problem of fair recommendation to a fair allocation problem, which leads to an interesting extension of previously studied problems owing to the specific constraints pertaining to recommendations (detailed in~\cref{mapping2}).

\vspace{1mm}
\noindent 
This paper is an extended version of our earlier work, titled ``FairRec: Two-Sided Fairness for Personalized Recommendations in Two-Sided Platforms"~\cite{patro2020fairrec}. 
In~\cite{patro2020fairrec}, we introduced the notions of two-sided fairness in recommendations and proposed the {\it FairRec} algorithm to ensure fairness for both producers and customers. In this work, we prove that FairRec ensures a stronger theoretical guarantee on the producer-side. Additionally, since we provide a tunable parameter $\alpha$ to the platforms to regulate the minimum exposure guarantee for the producers, we present a detailed analysis on what actually happens when the value of $\alpha$ is changed and its impact on both producer-side and customer-side. In certain scenarios, the platforms might be interested in ensuring different levels of exposure guarantee for different producers; for example,  the platform may want to give more exposure guarantee to higher-rated producers than the lower-rated ones. Thus, we also evaluate FairRec in such scenarios by tweaking the FairRec algorithm to entertain such requirements. We add several new baselines, such as MPB19~\cite{abdollahpouri2019managing}, MSR18~~\cite{surer2018multistakeholder}, MixedTR-$k$, MixedTP-$k$, to compare against FairRec. Finally, we propose a new modification of FairRec (named as FairRecPlus) that utilizes the envy graph to improve the recommendation performance for the customers, but comes at a cost of increased time complexity.
\section{Preliminaries}
\label{notions}
In this section, we define the terminology and notations used throughout the paper.

\subsection{Products and Producers}
In a few two-sided platforms focusing on physical establishments (e.g., Google Maps, Yelp), a producer typically owns one product (e.g., restaurant or shop); whereas in multimedia platforms like Spotify, YouTube or Netflix, an artist can produce multiple songs or videos; the same is also true for ecommerce platforms like Amazon and Flipkart, where one producer can list many products. To generalize our approach to both types of two-sided platforms, we consider products and producers to be equivalent, and use the terms `product' and `producer' interchangeably. Even for platforms where a producer can have multiple products, ensuring fairness at the product level can ensure fairness for individual producers -- where fairness can be ensured by making the exposure proportional to the producer's portfolio size.\footnote{In platforms with producers having multiple products, ensuring producer-side fairness other than proportionality remains open for future work.}

\subsection{Notations}
\label{notations}
Let $U$ and $P$ be the sets of customers and producers respectively, where $|U|=m$, and $|P|=n$.
Let $k$ be the number of products to be recommended to every customer. $R_u\subset P$ represents the set of $k$ products recommended to customer $u$; $|R_u|=k$.

\subsection{Relevance of Products}
\label{relevance}
The \textit{relevance} of a product $p$ to customer $u$, denoted as $V_u(p)$, represents the likelihood that $u$ would like the product $p$.
Formally, relevance is a function from the set of customers and products to the real numbers $V:~ U \times P \rightarrow \mathbb{R}$.
Usually, the relevance scores are predicted using various data-driven methods~(e.g., neighborhood-based methods~\cite{ning2015comprehensive}, latent factorization methods~\cite{liang2016factorization,koren2009matrix}, etc.), and $V_u(p)$ is a proxy for the utility gained by $u$ if product $p$ is recommended to her.

\subsection{Customer Utility}
\label{customer_utility}
The utility of a recommendation $R_u$ to a customer $u$ is proportional to the sum of relevance scores of products in $R_u$. 
Thus, recommending the $k$ most relevant products will give the maximum possible utility.
Let $R_u^*$ be the set of top-$k$ relevant products for $u$.
We use a normalized form of customer utility from $R_u$, defined as: $\phi^\text{}_u(R_u)=\frac{\sum_{p\in R_u} V_u(p)}{\sum_{p\in R_u^*} V_u(p)}$.

\subsection{Producer Exposure}
\label{exposure}
Exposure of a producer/product $p$ is the total amount of attention that $p$ receives from all the customers to whom $p$ has been recommended. In this paper, we assume a uniform attention model\footnote{There can be more sophisticated attention models considering \textit{position bias}~\cite{agarwal2019estimating}, where customers pay more attention to the top ranked products than the lower ranked ones. This being an initial work on two-sided fair recommendation (formulated as a fair-allocation problem), we focused on a basic model setting without position bias.} where customers pay similar attention to all $k$ recommended products, and express the exposure of a product $p$ 
as $E_p = \sum_{u\in U}\mathbbm{1}_{R_u}(p)$, where $\mathbbm{1}_{R_u}(p)$ is $1$ if $p\in R_u$, and $0$ otherwise. The sum of exposures of all the products is $\sum_{p\in P}E_p=m\times k$.
Note that we assume the relevance of a product does not play any role in producer's utility (in contrast to \citet{singh2018fairness,biega2018equity}), and use only the exposure of a producer as her utility.
\section{Need for two-sided fairness in personalized recommendations}
\label{motivation}
Traditionally, the goal of personalized recommendation has been to recommend products that would be most relevant to a customer. This task typically requires learning the relevance scoring functions ($V$), and several state-of-the-art data-driven methods~\cite{ning2015comprehensive,liang2016factorization,zhou2011functional,kim2016convolutional,xue2017deep} have been developed to estimate the product-customer relevance values. 
Once these values are obtained, the standard practice, across several recommender systems, is to recommend the top-$k$ ($k$=size of recommendation) relevant products to corresponding customers.
While this approach attempts to 
maximize the satisfaction of individual customers, it can adversely affect the producers in a two-sided platform, as we explore next.

\begin{figure}[t]
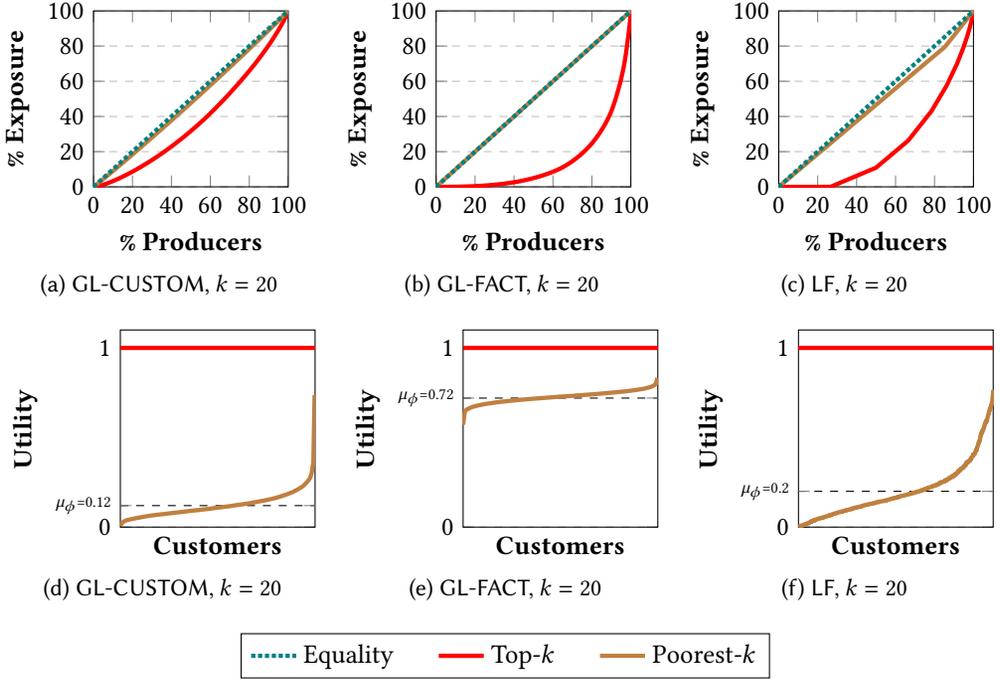

\center
{
		\subfloat[{GL-CUSTOM, $k=20$}]{\pgfplotsset{height=0.28\textwidth,width=0.3\textwidth,compat=1.9}\input{figures/lorentz_gl1_k_20}\label{fig:lorentz_gl1_k_20}}
		\hfil
		\subfloat[{GL-FACT, $k=20$}]{\pgfplotsset{height=0.28\textwidth,width=0.3\textwidth,compat=1.9}\input{figures/lorentz_gl2_k_20}\label{fig:lorentz_gl2_k_20}}
		\hfil
		\subfloat[{LF, $k=20$}]{\pgfplotsset{height=0.28\textwidth,width=0.3\textwidth,compat=1.9}\input{figures/lorentz_lf_k_20}\label{fig:lorentz_lf_k_20}}
		\vfil
		\subfloat[{GL-CUSTOM, $k=20$}]{\pgfplotsset{height=0.3\textwidth,width=0.3\textwidth,compat=1.9}\input{figures/cdf_gl1_k_20}\label{fig:cdf_gl1_k_20}}
		\hfil
		\subfloat[{GL-FACT, $k=20$}]{\pgfplotsset{height=0.3\textwidth,width=0.3\textwidth,compat=1.9}\input{figures/cdf_gl2_k_20}\label{fig:cdf_gl2_k_20}}
		\hfil
		\subfloat[{LF, $k=20$}]{\pgfplotsset{height=0.3\textwidth,width=0.3\textwidth,compat=1.9}\input{figures/cdf_lf_k_20}\label{fig:cdf_lf_k_20}}
		\vfil
		\subfloat{\pgfplotsset{width=.5\textwidth,compat=1.9}
			\begin{tikzpicture}
			\begin{customlegend}[legend entries={{Equality},{Top-$k$},{Poorest-$k$}},legend columns=3,legend style={/tikz/every even column/.append style={column sep=0.5cm}}]
			\addlegendimage{teal,ultra thick, densely dotted,sharp plot}
			\addlegendimage{red,mark=.,ultra thick,sharp plot}
			\addlegendimage{brown,mark=.,ultra thick,sharp plot}
			\end{customlegend}
			\end{tikzpicture}}
}
\caption{Lorenz curves (figures \ref{fig:lorentz_gl1_k_20},\ref{fig:lorentz_gl2_k_20},\ref{fig:lorentz_lf_k_20}) show high inequality among producer exposures with the top-k recommendation (using the ratio of rating to distance to estimate relevance in GL-CUSTOM, and using latent factorization model \cite{liang2016factorization} in GL-GACT and LF). 
In these curves, the cumulative fraction of total exposure is plotted against the cumulative fraction of the number of corresponding producers (ranked in increasing order of their exposures). The extent to which the curve goes below a straight diagonal line (or an equality mark) indicates the degree of inequality in the exposure distribution.
In figures \ref{fig:cdf_gl1_k_20}, \ref{fig:cdf_gl2_k_20}, and \ref{fig:cdf_lf_k_20}, we plot the individual customer utilities sorted in ascending order while $\mu_\phi$ represents the mean customer utility. While poorest-k provides almost equal exposures, it introduces huge loss and disparity in individual customer utilities (figures \ref{fig:cdf_gl1_k_20},\ref{fig:cdf_gl2_k_20},\ref{fig:cdf_lf_k_20}).}
\label{fig:lorenz_n_cdf}	
\end{figure}
\subsection{Datasets}
\label{dataset}
We consider the impact of customer-centric top-$k$ recommendations on the exposure of the producers using real-world datasets. We use a state-of-the-art relevance scoring model (a very widely used latent factorization method~\cite{liang2016factorization}) and also a dataset-specific custom relevance scoring model over the datasets. \\

\subsubsection{Google Local Ratings Dataset (GL)}~\\
Google Local is a service to find nearby 
shops, restaurants on Google Maps (as Google Nearby feature) platform.
We use the Google Local dataset released by \citet{he2017translation}, which contains data about customers, local businesses (producers), and their locations (geographic coordinates), ratings, etc.
We consider the active customers located in New York City and the business entities within $5$ miles radius of Manhattan area with at least $10$ reviews. The resulting dataset contains $11,172$ customers, $855$ businesses and $25,686$ reviews.
We consider the following two relevance scoring functions ($V$).\\

\begin{enumerate}
\item \textbf{\textit{GL-CUSTOM:}}
We use a custom relevance scoring function: $V(u,p)=\frac{rating(p)}{distance(u,p)}$, where $rating(p)$ is the average rating of the producer (local business) and $distance(u,p)$ is the distance between customer $u$ and producer $p$. 

\vspace{0.5mm}
\item \textbf{\textit{GL-FACT:}}
Here we use the state-of-the-art latent factorization model~\cite{liang2016factorization,koren2009matrix} to predict the relevance scores from the ratings. 
\end{enumerate}
From here on, we refer to the above two datasets (relevance score sets) as GL-CUSTOM and GL-FACT respectively.

\vspace{1mm}
\subsubsection{Last.fm Dataset (LF)}~\\
We use the {\tt Last.fm} dataset released by \citet{Cantador:RecSys2011}, which contains $1,892$ customers, $17,632$ artists (producers), and $92,834$ records of play counts (the number of times a customer has played songs from an artist).
We again use a latent factorization model~\cite{liang2016factorization,koren2009matrix} to find out the relevance scores from the play counts. 
From here on, we refer to this dataset as LF.

\subsection{Adverse Impact of Customer-Centric Recommendation}
\label{observations}
We simulate top-$k$ ($k=20$) recommendation on all three datasets, 
and calculate the exposure different producers get.
Figures-\ref{fig:lorentz_gl1_k_20},\ref{fig:lorentz_gl2_k_20},\ref{fig:lorentz_lf_k_20} are the Lorenz curves for producer exposures.
In Exposure Lorenz curves, the cumulative fraction of total exposure is plotted against the cumulative fraction of the number of corresponding producers (ranked in increasing order of their exposures). The extent to which the curve goes below a straight diagonal line (or an equality mark) indicates the degree of inequality in the exposure distribution.
We observe that the Lorenz curves for top-$k$ recommendations are far below the equal exposure marks, revealing that for conventional top-$k$ recommendation, $50$\% least exposed producers get 
only $32$\%, $5$\%, and $11$\% of total available exposure ($m\cdot k$) in {\em GL-CUSTOM} (using the ratio of rating to distance to estimate relevance), {\em GL-FACT} (using latent factorization model \cite{liang2016factorization}), and {\em LF} datasets (using latent factorization model \cite{liang2016factorization}), respectively. \\

\noindent {\bf Huge disparity is in nobody's interest: } In two-sided platforms, the exposure determines the economic opportunities.
Thus, low exposure on a platform often puts many producers at huge losses, forcing them to leave the platform;
this may result in fewer available choices for the customers, thereby degrading the overall quality of the platform.
Thus, highly skewed exposure distribution of the customer-centric top-$k$ recommendation not only makes it unfair to the producers but also questions the long-term sustainability of the platforms. 
Thus, there is a need to be fair to the producers while designing recommender systems.

\subsection{Pitfalls of Producer-Centric Solution}
A naive approach to reduce inequality in producer exposures is to implement a producer-centric recommendation (poorest-$k$):
recommend the least-$k$ exposed products to the customer at any instant. Such producer-centric scheme makes the exposure of all the producers nearly equal, as seen in figures-\ref{fig:lorentz_gl1_k_20},~\ref{fig:lorentz_gl2_k_20} and ~\ref{fig:lorentz_lf_k_20}: 
Lorenz curves for poorest-$k$ recommendations are closer to the diagonal than those of top-$k$ recommendations. 
However, such poorest-$k$ recommendation decreases overall customer utilities (as seen in figures-\ref{fig:cdf_gl1_k_20},~\ref{fig:cdf_gl2_k_20} and ~\ref{fig:cdf_lf_k_20}).
Moreover, the poorest-$k$ introduces disparity among individual customer utilities, where some customers may suffer much higher losses than other customers, making it unfair to them. 

\subsection{Desiderata of Fair Recommendation}
\label{fairness_properties}
To counter the above-mentioned issues, in this work, we 
propose the following fairness properties to be satisfied by the recommendation to be fair to both producers and customers. \\

\noindent \textbf{A. Producer Fairness: }
Mandating a uniform exposure distribution over the producers can be too harsh on the system;
it may heavily hamper the quality of the recommendation, and might also kill the existing competition by discouraging the producers from improving the quality of their products or services.
Instead, we propose to {\bf ensure a minimum exposure guarantee for every producer} such that no producer starves for exposure. The proposal is comparable to the fairness of {\it minimum wage guarantee} (e.g., as required by multiple legislations in the US, starting from Fair Labor Standards Act 1938 to Fair Minimum Wage Act 2007~\cite{pollin2008measure,green2010minimum,falk2006fairness}).
Ensuring minimum wage does not itself guarantee equality of income;
however it has been found to decrease income inequality~\cite{lin2016effects,engbom2018earnings}. Similarly, we want to ensure minimum exposure to every producer in the system. Note that we are not mandating a fixed exposure guarantee -- the exact value of the guaranteed minimum exposure 
can be decided by the respective platforms. \\

\noindent \textbf{B. Customer Fairness: }
As maintaining producer fairness can cause an overall loss in customer utility, we propose that {\bf the loss in utility  should be fairly distributed among the customers}. To ensure this, products need to be recommended in a way such that no customer can gain extra utility by exchanging her set of recommended products with another customer -- a property called {\it envy-freeness}, as detailed in the next section.
\section{Re-imagining Fair Recommendation as Fair Allocation}
\label{mapping}
Given a set of items (say, $\mathcal{P}$), a set of agents (say, $\mathcal{U}$), and valuations $\mathcal{V}$ (how much an agent values an item),
the \textbf{\textit{fair allocation problem}} aims at distributing the items \textit{fairly} among the agents. In the discrete version of this problem, the items are \textit{discrete} (no item can be broken into pieces) and \textit{non-shareable} (no item can be allocated to multiple agents).
If $\mathcal{P}$ contains several copies of the same item, 
each copy can be thought of as non-shareable and discrete.
The goal is to find 
a non-shareable and discrete allocation ($\mathcal{A}:=\{(A_u)_{u\in \mathcal{U}}: A_u\subseteq \mathcal{P}\}$) while ensuring 
fairness properties.

\subsection{Notions of Fairness in Allocation}
The classical fairness notions, such as envy-freeness\footnote{An allocation is said to satisfy \textit{envy-freeness} if the bundle of items allocated to each agent is as valuable to her as the bundle allocated to any other agent~\cite{foley1967resource,varian1974equity,stromquist1980cut}.} (EF) and proportional-fair-share\footnote{An allocation is said to satisfy \textit{proportional-fair-share} if each agent receives a bundle of value at least $1/|\mathcal{U}|^{th}$ of her total value for all the items~\cite{steinhaus1948problem}.} (PFS), may not be achievable in most instances of the problem.
For example, if there are two agents and one item, the item will be allocated to one of the agents, and the zero allocation to the other agent would violate both EF and PFS.
Thus, for \textit{discrete} items, relaxations of EF and PFS have been considered.
Two such well-studied notions of fairness in the discrete fair allocation literature are (i)~\emph{envy freeness up to one item} ($\EFone$) and (ii)~the \textit{maximin share guarantee} ($\MMS$), defined by~\citet{budish2011combinatorial}. Since then, these have been extensively studied in various settings for providing existential and algorithmic guarantees~\cite{amanatidis2018comparing,barman2018groupwise,brandt2016handbook,amanatidis2015approximation,procaccia2014fair,kurokawa2016can,bouveret2014characterizing,caragiannis2016unreasonable,biswas2018fair,biswas2019matroid,fair-graph,bilo2018almost}.
We now formally state these fairness notions:
\begin{itemize}[leftmargin=*]
	\item An allocation $\mathcal{A}$ is $\EFone$ iff for every pair of agents $u , w \in \mathcal{U}$ there exists an item $p \in A_w$ such that $\mathcal{V}_u(A_u) \geq \mathcal{V}_u(A_w\setminus\{p\})$.
	
	\item An allocation is said to satisfy $\MMS$ if each agent receives a value greater than or equal to their \emph{maximin share} threshold. This threshold for an agent $u$ is defined as
	\begin{equation}
		\MMS_u= \max_{\mathcal{A}}\  \min_{w\in \mathcal{U}}\ \mathcal{V}_u(A_w).\label{eq:mms}
	\end{equation}
	In other words, $\MMS_u$ is the maximum value that the agent can guarantee for herself if she were to allocate $\mathcal{P}$ into $|\mathcal{U}|$ bundles and then, from those bundles, receive the minimum valued one.
	Formally, an allocation $\mathcal{A}$ satisfy $\MMS$ fairness iff for all agents $u \in \mathcal{U}$, we have $\mathcal{V}_u(A_u) \geq \MMS_u$.
\end{itemize}

\subsection{Fair Recommendation to Fair Allocation}
We propose to see the desired two-sided recommendation problem as a {\em fair allocation} problem.
The set of products $P$ can be thought of as the set of items $\mathcal{P}$ (there can be multiple copies of individual products)\footnote{Note that $P$ represents the set of producers or products. On the other hand, $\mathcal{P}$ represents the set of items that are to be allocated among the customers; this set of items can be suitably formed by creating and gathering copies of each product based on how much exposure we want to guarantee for that product.};
similarly, the set of customers $U$ as the set of agents $\mathcal{U}$, and the relevance scoring function $V$ as the valuations $\mathcal{V}$. Now the task of recommending products to customers is the same as allocating items in $\mathcal{P}$ to agents in $\mathcal{U}$ with certain constraints.

\begin{itemize} [leftmargin=*]
\item {\bf Setting $\mathcal{P}$ for Producer Fairness:} As the total exposure of the platform 
is limited ($k\cdot |U|$), the maximum guarantee on minimum possible exposure for the producers is $\left\lfloor\frac{k\cdot |U|}{|P|}\right\rfloor$ (this refers to the $\mathrm{MMS}$ value for the producers). 
One way to formally define a lower threshold requirement $\overline{E}$ is by using the notion of \textit{approximate maximin share} ($\alpha$-$\MMS$). More formally,
 
\begin{definition}
An allocation $\mathcal{A}=(A_1,\ldots, A_n)$ is said to satisfy $\alpha$-$\MMS$ for a fixed value $\alpha\in(0,1]$ if and only if, for all agents $i \in\{1,\ldots, n\}$, the following holds:

\[\mathcal{V}_i(A_i)\geq \alpha\MMS_i, \quad\quad\mbox{ where }\quad \MMS_i = \underset{\mathcal{A}}{\max}\ \ \ \underset{j\in\{1,\ldots,n\}}{\mathrm{argmin}}\ \ \  \mathcal{V}_i(A_j).\]
\end{definition}
A platform can decide the exact value of $\alpha$ and provide the minimum exposure guarantee $\overline{E} = \alpha \mathrm{MMS}$ for every producer. \\

\item {\bf Fair Allocation of $\mathcal{P}$ among $\mathcal{U}$:}
Once $\mathcal{P}$ is set according to the desired producer fairness, the entire task of fair recommendation boils down to the allocation of $\mathcal{P}$ among $\mathcal{U}$ while ensuring $\EFone$ for agents/customers (to ensure fairness for customers, as introduced in \cref{fairness_properties}). However, specific constraints related to the recommendation problem need a novel extension of the traditional fair allocation problem, as explained next.
\end{itemize}

\subsection{Extending the Conventional Fair Allocation Problem}\label{mapping2}
Traditionally, \textit{fair allocation} literature aims at defining and ensuring \textit{fairness} among the agents while allocating all items 
the set $\mathcal{P}$ exhaustively. However, in the {\it fair recommendation} problem, along with customer fairness, the challenge is to attain producer (or product) fairness by providing a minimum exposure guarantee (where each product needs to be allocated to at least $\ell$ different customers). Thus, achieving producer fairness is the same as creating at least $\ell$ copies of each product and ensuring that all the copies are allocated, along with a feasibility constraint which enforces that no customer gets more than one copy of the same product. This extension of the problem---where all the items are grouped into disjoint categories and no agent receives more than a pre-specified number of items from the same category---is called cardinality constrained fair allocation problem, proposed 
in~\cite{biswas2018fair}. In this paper, we consider a novel extension of the cardinality constrained problem by adding another constraint enforcing that exactly $k$ items are allocated to each customer. This requires tackling hierarchical feasibility constraints---an upper bound cardinality constraint of one on each product and a cardinality constraint of $k$ on the total number of allocated products. Moreover, this additional feasibility constraint makes it difficult to decide how many copies of which product should be made available for a total of ($k\cdot |U|$) allocations, satisfying the feasibility constraints as well as the fairness requirement. 
Thus, unlike the fair allocation problem, we consider no restriction on the number of copies of each product that are made available. All these contrast points, along with the two-sided fairness guarantees make fair recommendation an interesting extension of the fair allocation problem.

Overall, we aim to recommend a set of $k$ products, denoted as $A_u\subset P$, to each customer $u\in U$ that satisfies the following constraints:

\begin{eqnarray}
\displaystyle \sum_{p\in A_u} V_u(p)&\geq & \displaystyle \sum_{p\in A_w} V_u(p) - \underset{p\in A_w}{\max}\ V_u(p) \mbox{ for every pair of customers }u,w\in\mathcal{U}.\\
|A_u| &=& k, \mbox{ for all customers } u \in U.\\
\displaystyle \sum_{p\in P} E_p &\geq & \left\lfloor\frac{\alpha m k}{n}\right\rfloor, \mbox{ for each producer }p\in P, \mbox{ for some }\alpha \in (0,1].
\end{eqnarray}

This formulation leads to a large number of constraints, of the order $|U|^2+|P|$. Standard Integer Linear Programming techniques for finding a feasible solution, in this large constraint space, do not scale well, especially considering the real-world two-sided platforms. Hence, in the next section, we propose a greedy algorithm to solve the above-mentioned constraint satisfaction problem.
\section{FairRec: An Algorithm to ensure two-sided fairness}
\label{algorithm}
In this section, we provide a polynomial-time algorithm {\em FairRec}, for finding an allocation $\mathcal{A}$ which satisfies the desired two-sided fairness described in \cref{fairness_properties} (we prove the theoretical guarantees in ~\cref{sec:theorems}). Note that we consider only the case of $k<|P|$, and leave the trivial case of $k=|P|$ and the infeasible case of $k>|P|$ out of consideration.
Also, we consider $|P|\leq k\cdot |U|$, otherwise, at least ($|P|-k\cdot |U|$) producers can not be allocated to any customer.

\begin{algorithm}[t]
		{\raggedright
			{
				{\raggedright{\bf Input:} Set of customers $U=[m]$, set of distinct products $P=[n]$,   recommendation set size $k$ (such that $k<n$ and $n\leq k\cdot m$), and the relevance scores $V_u(p)$.}\\
				{\raggedright{\bf Output:} A two-sided fair recommendation.}
			}
			\caption{{\em FairRec} ($U, P, k, V$)}
			\label{alg:two-sided}
			\begin{algorithmic}[1]
				
				\State Initialize allocation $\mathcal{A}^0=(A^0_1, \ldots, A^0_m)$ with $A^0_i~\leftarrow~\emptyset$ for each customer $i \in [m]$.
				
				\NoNumber{}
				\NoNumber{\textbf{First Phase:}}
				\State Fix an (arbitrary) ordering of the customers $\sigma = \left(\sigma(1), \sigma(2),\ldots, \sigma(m)\right)$.
				\State Initialize set of feasible products $F_u \leftarrow P$ for each $u\in [m]$.
				\State Set $\ell \leftarrow \left\lfloor \frac{\alpha m k}{n}\right\rfloor$ denoting number of copies of each product.	
				\State Initialize each component of the vector $S=(S_1, \ldots, S_n)$ with $S_j\leftarrow \ell$, $\forall j\in [n]$, this stores the number of available copies of each product.
				\State Set $T\leftarrow \ell\times n$, total number of items to be allocated.
				\State $[\mathcal{B}, F, x]\leftarrow$Greedy-Round-Robin$(m,n,S,T,V,\sigma,F)$.
				\State Assign $\mathcal{A}\leftarrow \mathcal{A}\cup \mathcal{B}$.
				\NoNumber{}
				\NoNumber{\textbf{Second Phase:}}
				\State Set $\Lambda=|A_{\sigma((x)\Mod{m}+1)}|$ denoting the number of items allocated to the customer subsequent to $x$, according to the ordering $\sigma$. 
				\If{ $\Lambda<k$}
				\State Update each component of the vector $S=(S_1,\ldots,S_n)$ with the value $m$ in order to allow allocating any product to any customer.
				\State Set $T\leftarrow 0$.
				\If{ $x<m$ }
				\State Set $\sigma'(i)\leftarrow \sigma((i+x-1)\Mod{m}+1)$ for all $i\in[m]$.
				\State $\sigma\leftarrow \sigma'$.
				\State $T\leftarrow (m-x)$.
				\State Update $\Lambda\leftarrow \Lambda+1$.
				\EndIf				
				\State $T\leftarrow T+m(k-\Lambda)$ total number of items to be allocated.
				\State $[\mathcal{C},F,x]\leftarrow$Greedy-Round-Robin$(m,n,S,T,V,\sigma,F)$.
				\State Assign $\mathcal{A}\leftarrow \mathcal{A}\cup \mathcal{C}$.
				\EndIf
				
				\State Return $\mathcal{A}$.
			\end{algorithmic}}
\end{algorithm}

\textit{FairRec} (Algorithm~\ref{alg:two-sided}) executes in two phases.
The first phase ensures $\EFone$ among all the $m$ customers (Lemma \ref{lemma:customer_fairness}) and tries to provide a minimum guarantee on the exposure of the producers (Lemma~\ref{lemma:producer_fairness}).
However, the first phase may not allocate exactly $k$ products to all the $m$ users, which is then ensured by the second phase while simultaneously maintaining $\EFone$ for customers. 
	
	The {\bf first phase} creates $\ell=\left\lfloor\frac{\alpha mk}{n}\right\rfloor$ copies of each product. Note that $\left\lfloor\frac{mk}{n}\right\rfloor$ is the maximin value of any producer 
	 when $mk$ slots are allocated among $n$ producers, and thus $\ell$ represents the $\alpha$-MMS value for each producer. 
	The algorithm then initializes each component of the vector $S$ of size $|P|$ to $\ell$ to ensure that at most $\ell$ copies from each product are allocated in the first phase.
	Feasible sets $F_u$ for each customer $u$ are then initialized to ensure that each customer receives at most one copy of the same product.
	Then, assuming an arbitrary ordering $\sigma$ of customers, Algorithm~$\ref{alg:greedy}$ is executed and the allocation $\mathcal{B}$ is obtained.

	The {\bf second phase} checks if all the customers have received exactly $k$ products (by looking at the number of products allocated to the customer $x+1$ which is next-in-sequence to the last allocated customer $x$ of the first phase). If the customer $x+1$ has received $k$ products, then no further allocation is required;
	if not, then Algorithm~\ref{alg:greedy} is called again with a new ordering obtained by $x$ left-cyclic rotations of $\sigma$.
	The remaining number of items is stored in $T$ which are to be allocated among the customers.
	Also, each component of the vector $S$ is updated to $|U|$ to allow allocating any feasible product without any limit on the available number of copies.
	The second phase retains $\EFone$ fairness among the customers. 
	
	Both phases use a modified version of the Greedy-round-robin (Algorithm~\ref{alg:greedy})~\cite{caragiannis2016unreasonable,biswas2018fair}:
	it follows the ordering $\sigma$ in a round-robin fashion (i.e., it selects customers, one after the other, from $\sigma(1)$ to $\sigma(m)$), and iteratively assigns to the selected customer her most desired unallocated product (feasibility maintained by the vector $S$ and sets $F_u$ and ties are broken arbitrarily).
	This process is repeated over several rounds until one of the two disjoint conditions occur: (i)~$T==0$: a total of $T$ allocations have occurred, or (ii)~$p==\emptyset$: no feasible product available (for the current customer $\sigma(i)$, we have $F_{\sigma(i)} \cap \{p: S_p\neq 0\}= \emptyset$). Finally, it returns an allocation $B_1,\ldots,B_m$ with each $B_u\subseteq[n]$ for all $u\in[m]$.
	
Note that while we choose minimum exposure guarantee and EF1 as fairness constraints for producers and customers respectively, it does not necessarily mean that merely satisfying them is enough.
Apart from ensuring fairness, we also need to perform well in what any recommender system is originally designed to do, i.e., to provide good (relevant) recommendations to the customers.
Thus, in Algorithm~\ref{alg:greedy} which allocates products to customers in greedy-round-robin manner, we try to allocate (in step-7) the best product which is available and feasible in every round.
We believe this is the reason why FairRec shows good performance in overall customer utility as well (based on experimental results in \cref{experiments}).
\begin{algorithm}[t]
			{\raggedright
				{
					{\raggedright {\bf Input :} Number of customers $m$, number of producers $n$, an array with number of available copies of each product $S$, total number of available products $T>0$, relevance scores $V_u(p)$ and feasible product set $F_u$ for each customer, and an ordering $\sigma$ of $[m]$.}\\
					{\raggedright {\bf Output:} An allocation of $T$ products among $m$ customers, the residual feasible set $F_u$ and the last allocated index $x$.
					}
					\caption{Greedy-Round-Robin ($m,n,S,T,V,\sigma,F$)}
					\label{alg:greedy}
					\begin{algorithmic}[1]
						\State Initialize allocation $\mathcal{B}=(B_1, \ldots, B_m)$ with $B_i~\leftarrow~\emptyset$ for each customer $i \in [m]$.	
						\State Initiate $x \leftarrow m$.
						\State Initiate round $r\leftarrow 0$.
						\While{ true }
						\State Set $r\leftarrow r+1.$
						\For{$i =1 \mbox{ to } m$}
						\State Set $p \in \underset{p'\in F_{\sigma(i)}:(S_p\neq 0)}{\argmax} V_{\sigma(i)}(p')$ 
						\If{$p == \emptyset$}
						\State Set $x=i-1$ only if $i\neq 1$.
						\State \textbf{go to} Step $22$.
						\EndIf
						\State Update $B_{\sigma(i)} \leftarrow B_{\sigma(i)} \cup p$. 
						\State Update $F_{\sigma(i)} \leftarrow F_{\sigma(i)} \setminus p$.
						\State Update $S_p\leftarrow S_p-1$.
						\State Update $T\leftarrow T-1$.					
						\If{ $T==0$ }
						\State $x=i$.
						\State \textbf{go to} Step $22$.
						\EndIf
						\EndFor
						\EndWhile			
						\State Return $\mathcal{B}=(B_1,\ldots,B_m)$, $F=(F_1,\ldots,F_m)$ and index $x$.
					\end{algorithmic}}}
				\end{algorithm}			
\section{Theoretical Guarantees}\label{sec:theorems}
In this section, we provide a few important properties of Algorithm~\ref{alg:greedy} in Proposition~\ref{proposition:greedy}. Later, we establish the fairness guarantees and time complexity of our proposed algorithm \textit{FairRec} in Theorem~\ref{theorem:two-sided} using Lemma~\ref{lemma:customer_fairness}, \ref{lemma:alpha-producers} and \ref{lemma:polytime}.

\begin{proposition}\label{proposition:greedy}
	The allocation obtained by the $\mathrm{Greedy}$-$\mathrm{Round}$-$\mathrm{Robin}$ (Algorithm~\ref{alg:greedy}) exhibits the following four properties:
	\begin{enumerate}[label=(\subscript{P}{{\arabic*}})]
		\item for any two indices $x$ and $y$, where $x < y$, the customer $\sigma(x)$ (who appears earlier than $\sigma(y)$ according to the ordering $\sigma$) does not envy customer $\sigma(y)$, i.e., $V_{\sigma(x)} (B_{\sigma(x)} ) \geq V_{\sigma(x)} (B_{\sigma(y)})$.
		\item the allocation $\mathcal{B}$ obtained by Algorithm~\ref{alg:greedy} is $\EFone$.
		\item each customer is allocated at most one item from the same producer, thus ensuring the cardinality constraint is satisfied for each producer (category).
		\item for any two customers, say $u$ and $w$, the allocation $\mathcal{B}$ obtained by Algorithm~\ref{alg:greedy} satisfies the following: $-1\leq \left(|B_u|-|B_w|\right) \leq 1$.    
	\end{enumerate}
\end{proposition}

\begin{proof}
	The properties $P_1$ and $P_2$ have been observed by ~\citet{biswas2018fair} and ~\citet{caragiannis2016unreasonable}, respectively. For completeness, we repeat the arguments towards these two properties. Let $x$ and $y$ be two indices, such that $1\leq x<y\leq m$. At each round $r$, the customer $\sigma(x)$ chooses her most desired product among all the unallocated items before customer $\sigma(y)$. Hence, $V_{\sigma(x)}(p^r_{\sigma(x)})~\geq~V_{\sigma(x)}(p^r_{\sigma(y)})$, where $p^r_{\sigma(x)}$ and $p^r_{\sigma(y)}$ denote the items assigned to customer $\sigma(x)$ and $\sigma(y)$, respectively. Thus, over all the rounds,  $\sum_r V_{\sigma(x)}(p^r_{\sigma(x)}) \geq \sum_r V_{\sigma(x)}(p^r_{\sigma(y)})$ which implies that $V_{\sigma(x)}(B_{\sigma_{x}})\geq V_{\sigma(x)}(B_{\sigma_{y}})$ and thus the property $P_1$ holds. 
	
	Property $P_2$ states that if $\sigma(y)$ envies $\sigma(x)$, it will not violate $\EFone$ property (note: we already saw in $P_1$ that $\sigma(x)$ does not envy $\sigma(y)$). Now observe that, the value $V_{\sigma(y)}(\cdot)$ of the item allocated to customer $\sigma(y)$ in the $r$th round is at least that of the item allocated to customer $\sigma(x)$ in the $(r+1)$th round. Let $R$ denote the total number of rounds, then the following holds: 
	\begin{align} 
		&V_{\sigma(y)}(p^r_{\sigma(y)}) \geq V_{\sigma(y)}(p^{r+1}_{\sigma(x)})\quad \mbox{ for all } r \in \{1,\ldots,R-1\}\nonumber \\ 
		\Rightarrow &\sum_{r=1}^{R-1} V_{\sigma(y)}(p^r_{\sigma(y)}) \geq \sum_{r=1}^{R-1} V_{\sigma(y)}(p^{r+1}_{\sigma(x)}) \nonumber\\
		\Rightarrow &V_{\sigma(y)}(B_{\sigma(y)}) \geq V_{\sigma(y)}(B_{\sigma(x)}) -V_{\sigma(y)}(p^1_{\sigma(x)})\label{eq:EF1_greedy}
	\end{align}
	Equation~\ref{eq:EF1_greedy} shows that the customer $\sigma(y)$ stops envying $\sigma(x)$ when only one item is (hypothetically) removed from $\mathcal{B}_{\sigma(x)}$ (namely, $p^1_{\sigma(x)}$). Thus, the allocation $\mathcal{B}$ is $\EFone$, i.e., $P_2$ holds.
	
	The property $P_3$ is satisfied by the use of the feasible sets $F_u$ for each customer $u$. Each $F_u$ contains the set of producers who have not yet been allocated to the customer $u$. At any round $r$, step $7$ of Algorithm~\ref{alg:greedy} selects the most relevant producer among the producers who had not been allocated to $u$ in any earlier rounds $r'<r$. Once, a producer $p$ is allocated to a customer $u$, step $9$ of Algorithm~\ref{alg:greedy} removes $p$ from $F_u$. Thus, each customer is allocated at most one item from the same producer.
	
	The property $P_4$ states that, for any pair of customers $u$ and $w$, the number of allocated items $|B(u)|$ and $|B(v)|$, differ by at most $1$. It is straightforward to see that, except for the last feasible round, all customers are allocated exactly one item at each round. Thus, all the customers receive the same number of allocations until the second last feasible round. In the last feasible round, some customers may not get any allocation (if there is no available feasible product) and thus may receive one item less than the others.
\end{proof}

\noindent We now state the main theorem (Theorem~\ref{theorem:two-sided}) that establishes the fairness guarantees of our proposed algorithm.

\begin{theorem}\label{theorem:two-sided}
	Given $n$ producers, the proposed polynomial time algorithm, $\mathrm{FairRec}$, returns an $\EFone$ allocation among $m$ customers while allocating exactly $k$ items  to each customer, when $k<n\leq mk$. Moreover, it ensures non-zero exposure among all the $n$ producers and $\alpha$-$\MMS$ guarantee among at least $\left(1- \frac{1}{m+1}\left\lfloor\frac{\alpha m k}{n}\right\rfloor\right)$ fraction of the producers. 
\end{theorem}

\begin{proof}
	We prove the fairness guarantees of customers and producers in Lemma~\ref{lemma:customer_fairness} and \ref{lemma:alpha-producers}, respectively. In Lemma~\ref{lemma:polytime}, we show that $\mathrm{FairRec}$ executes in polynomial time.
\end{proof}

\begin{lemma}\label{lemma:customer_fairness}
	Given $n$ producers, $m$ customers, and a positive integer $k$ (such that $k<n\leq mk$), $\mathrm{FairRec}$ returns an $\EFone$ allocation among $m$ customers while allocating exactly $k$ items to each customer.
\end{lemma}

\begin{proof}
	To prove this, we show that both phases of $\mathrm{FairRec}$ satisfy $\EFone$. Since Algorithm~\ref{alg:greedy} guarantees $\EFone$ (by property $P_2$), the allocation $\mathcal{A}$ at step $9$ of $\mathrm{FairRec}$ is $\EFone$. Thus, for any two customers $u$ and $w$, there exists an item $j\in B_w$ such that $V_u(B_u)\geq V_u(B_w)-V_w(j)$. Next, the second phase creates $|U|$ copies of each product and calls Algorithm~\ref{alg:greedy} to obtain the allocation $\mathcal{C}$. Note that the second phase assigns the most valued item to each customer at each round, that is, it allocates top-$\Lambda_u$ feasible producers to each customer, where $\Lambda_u = k-|B_u|$. Thus, $V_u(C_u)\geq V_u(C_w)$. Thus,  $V_u(B_u)+V_u(C_u)\geq V_u(B_w)-V_w(j) + V_u(C_w)$, which implies $\EFone$: $V_u(B_u\cup C_u)\geq V_u(B_w\cup C_w)-V_u(j)$. This completes the proof that $\mathrm{FairRec}$ ensures $\EFone$ among all the customers while recommending exactly $k$ products to each customer.
\end{proof}

We now establish the fairness guarantees of \textit{FairRec} for the producers. For ease of exposition, we first consider $\alpha=1$ to show exact $\MMS$ fairness guarantees are satisfied by at least $(1-\frac{n}{k})$ fraction of the producers (in Lemma~\ref{lemma:producer_fairness}). Subsequently, we provide a stronger guarantee in Lemma~\ref{lemma:alpha-producers}  considering any $\alpha\in(0,1]$. 

\begin{lemma}\label{lemma:producer_fairness}
	Given $n$ producers, $m$ customers, a positive integer $k$ (such that $k<n\leq mk$), and $\alpha=1$, $\mathrm{FairRec}$ ensures non-zero exposure among all the $n$ producers. Moreover, it assures $\MMS$-fairness among at least $n-k$ producers. \end{lemma}

\begin{proof}
	We first prove that the first phase guarantees non-zero exposure for producers. The allocation $\mathcal{B}$ obtained by Algorithm~\ref{alg:greedy} in the first phase may have terminated for one of the two conditions
	\begin{enumerate}[leftmargin=*]
		\item $T==0$: this means that all the $\ell=\lfloor\frac{mk}{n}\rfloor$ copies of each producer have been allocated among all the customers. Thus, each producer receives exactly the maximin threshold $\ell$. Hence $\MMS$ fairness is achieved by all the $n$ producers. 
		\item $p==\emptyset$: this happens when $T\neq 0$ and $\sum_{p\in F_u} S_p = 0$ for a customer $u$ (at termination). That is, $S_p=0$ for each producer $p\in F_u$. Thus, all $\ell=\lfloor\frac{mk}{n}\rfloor$ copies of the producers in the set $F_u$ have been allocated, and hence they attain $\MMS$ fairness. On the other hand, the producers in the set $B_u$ (the set recommended to customer $u$) is allocated to at least one producer. Thus, a minimum value of $1$ is achieved by all the $n$ producers. Also, $|F_u|+|B_u|=n$ and $|B_u|\leq k$, implies that $|F_u|\geq n-k$. Therefore, at least $n-k$ producers attain $\MMS$-fairness.
	\end{enumerate}
	
	Since the thresholds are already satisfied in the first phase, adding more allocations in the second phase retains the threshold-based fairness guarantees. This completes the proof that $\mathrm{FairRec}$ ensures a non-zero exposure among all the $m$ producers and assures $\MMS$-fairness among at least $n-k$ producers.
\end{proof}

One consequence of Lemma~\ref{lemma:producer_fairness} is that, when $k$ is much lower than $n$, a large fraction of producers are guaranteed to attain $\MMS$ fairness. We formally state this property of $\mathrm{FairRec}$ algorithm in Corollary~\ref{lemma:corollary}.

\begin{corollary}\label{lemma:corollary}
	Given $n$ producers, a positive integer $k$, $\alpha=1$, and $\beta\in(0,1)$ such that $k\leq\beta n$, $\mathrm{FairRec}$ ensures $\MMS$-fairness among at least $(1-\beta)n$ producers.
\end{corollary}

\begin{lemma}\label{lemma:alpha-producers}
Given $n$ producers, $m$ customers, recommendation size $k$ such that $k<n<mk$, and a fixed value $\alpha\in(0,1]$, \textit{FairRec} ensures a minimum exposure of $1$ for all the producers and an $\alpha$-$\MMS$ guarantee to at least $\left(1- \frac{1}{m+1}\left\lfloor\frac{\alpha m k}{n}\right\rfloor\right)$ fraction of the producers.
\end{lemma}
\begin{proof}
The first phase of  \textit{FairRec} terminates with one of the two (mutually exclusive) options, either $T==0$ or $p== \emptyset$ (line~$16$ and line $8$ of Algorithm~\ref{alg:greedy}). The case $T==0$ ensures that all the producers achieve $\alpha$-$\MMS$ exposure guarantee. However, in the other case, when $p==\emptyset$, some producers may not achieve $\alpha$-$\MMS$ guarantee. We now provide a lower bound on the fraction of producers who achieve an $\alpha$-$\MMS$ guarantee after the first phase. 

Let $R$ be the number of rounds after which the first phase of the algorithm \textit{FairRec} terminates, and $Q$ be 
the total number of allocations that occurred in the first phase. Therefore,
\begin{equation}
Q\geq mR \label{eq:Qlow}
\end{equation}

Also, let $\beta$ be the fraction of producers who achieved $\alpha$-$\MMS$ exposure guarantee after the first phase. Thus, the total number of allocation $Q$ can be upper bounded as follows:
\begin{eqnarray}
Q & \leq & \beta n \left\lfloor\frac{\alpha m k}{n}\right\rfloor + (1-\beta) n \left(\left\lfloor\frac{\alpha m k}{n}\right\rfloor -1 \right)\nonumber \\
& \leq & n\left(\left\lfloor\frac{\alpha m k}{n}\right\rfloor + \beta -1 \right)\label{eq:Qup}
\end{eqnarray}

Combining Inequalities~\ref{eq:Qlow} and \ref{eq:Qup}, we obtain,
\begin{eqnarray}
mR & \leq & n\left(\left\lfloor\frac{\alpha m k}{n}\right\rfloor + \beta -1 \right)\nonumber\\
\beta & \geq & \frac{mR}{n} - \left\lfloor\frac{\alpha m k}{n}\right\rfloor + 1 \label{eq:beta_low}
\end{eqnarray}

Moreover, using the fact that at least $n-R$ producers achieve $\alpha$-$\MMS$ exposure guarantee, we obtain,
\begin{eqnarray}
\beta n & \geq & n - R\nonumber\\
\Rightarrow R & \geq & (1-\beta)n \label{eq:R_geq}
\end{eqnarray}

Using the lower bound on $R$ (Inequality~\ref{eq:R_geq}), the Inequality~\ref{eq:beta_low} can be rewritten as:
\begin{eqnarray}
\beta & \geq & m(1-\beta) - \left\lfloor\frac{\alpha m k}{n}\right\rfloor + 1\nonumber\\
\Rightarrow \beta (1+m)& \geq & m + 1 -  \left\lfloor\frac{\alpha m k}{n}\right\rfloor \nonumber\\
\Rightarrow \beta & \geq & 1 - \left\lfloor\frac{\alpha m k}{n}\right\rfloor\left(\frac{1}{m+1}\right)\label{eq:beta_final}
\end{eqnarray}

The Inequality \ref{eq:beta_final} implies that the fraction of producers who achieve $\alpha$-$\MMS$ guarantee is at least $1- \frac{1}{m+1}\left\lfloor\frac{\alpha m k}{n}\right\rfloor$. 
\end{proof}

Finally, in Lemma~\ref{lemma:polytime}, we show that $\mathrm{FairRec}$ executes in polynomial time.

\begin{lemma}\label{lemma:polytime}
	The time complexity of $\mathrm{FairRec}$ has a worst case bound of $\mathcal{O}(mnk)$.
\end{lemma}
\begin{proof}
	The time complexity of $\mathrm{FairRec}$ is the same as the complexity of Algorithm~\ref{alg:greedy}. Over the two phases, Algorithm~\ref{alg:greedy} allocates $mk$ items. For each allocation, it finds the maximum possible feasible producer which can be done in at most $\mathcal{O}(n)$ time. Thus, the total time complexity of the algorithm is $\mathcal{O}(mnk)$.
\end{proof}



%

\section{experimental evaluation}\label{experiments}
{\bf Experimental Setup and Baselines:}
We run the proposed {\em FairRec} algorithm (\cref{algorithm}) for all three datasets (as listed in \cref{dataset}) considering different values of the recommendation-size $k$.
For comparison, we use the following methods as baselines.\\
{\bf (1) Top-$k$:} This is the traditional way of recommending the top-$k$ relevant products.\\
{\bf (2) Random-$k$:} Here, we randomly recommend $k$ products to all customers. Random recommendations can give equal chance to all producers, thus can serve as a baseline which has only producers' interest in mind.\\
{\bf (3) Poorest-$k$:} Unlike random-$k$, this is a deterministic producer-centric method where $k$ least exposed products are recommended to each customer in a round robin manner.\\
Note that poorest-k and random-k are not real recommendation algorithms, and we consider them as baselines in the paper because of their theoretical property of bringing down inequality among producers and serve as some of the most equitable options from the producers' perspective.\\
{\bf (4) MixedTR-$k$:} Here, we choose top $\left\lceil\frac{k}{2}\right\rceil$ relevant products at first and then the remaining $\big(k-\left\lceil\frac{k}{2}\right\rceil\big)$ randomly, thus making it a mix of top-$\frac{k}{2}$ and random-$\frac{k}{2}$. In mixedTR-$k$, the top-$\frac{k}{2}$ could help keep high customer utility while random-$\frac{k}{2}$ could help in improving provider-side performance by giving equal chances to all producers to appear in the second half.\\
{\bf (5) MixedTP-$k$:} Here, we choose top $\left\lceil\frac{k}{2}\right\rceil$ relevant products at first and then the poorest $\big(k-\left\lceil\frac{k}{2}\right\rceil\big)$ producers, thus making it a mix of top-$\frac{k}{2}$ and poorest-$\frac{k}{2}$.
In mixedTP-$k$, the top-$\frac{k}{2}$ could help keep high customer utility while poorest-$\frac{k}{2}$ could help in improving the exposure of under-exposed providers by giving them a chance to appear in the second half.\\
{\bf (6) MPB19:} We use the method proposed by \citet{abdollahpouri2019managing} as a baseline here. The proposal by \citet{abdollahpouri2019managing} is to consider the relevance scores as intermediate scores and then add a benefit term for less exposed producers to promote diversity thereby reducing popularity bias. In our problem setting, we consider the relevance score of product $p$ to customer $u$ $V_u(p)$ as intermediate scores and express the modified relevance score as $0.5\times V_u(p)+0.5\times \big(1-\frac{E_p}{\sum_{p'\in P}E_{p'}}\big)$; the second part of the modified score promotes less exposed producers.\\
{\bf (7) MSR18:} \citet{surer2018multistakeholder} proposed to introduce producer-side constraints similar to exposure guarantees in our paper and then optimize overall customer utility. However, as the proposed constrained optimization problem becomes a very hard combinatorial problem, the authors did a Lagrangian relaxation and proposed to use iterative subgradient method to optimize the relaxed problem. Although this methodology is not quite suitable for large scale online platforms as optimizing the hard combinatorial problem could demand huge computing resources and time, we use this as a baseline by limiting the number of iterations to $100$.

~\\{\bf Experiments:} We run three sets of experiments. First in \cref{mms_experiments}, we set the exposure guarantee as $\overline{E}=$MMS and run FairRec.
However, the platforms may not always want to ensure the maximum possible exposure guarantee for the producers as it might cause degradation of customer utility.
They might want to set a lower exposure guarantee in such cases.
Thus, we set lower exposure guarantees i.e., by considering $\overline{E}=\alpha \cdot$MMS where $0 \leq \alpha \leq 1$ (in \cref{alpha_mms_experiments}). 
While \cref{mms_experiments} and \cref{alpha_mms_experiments} show the efficacies of {\em FairRec}, \cref{subsec:phase_interplay} digs deeper into the functioning of {\em FairRec} and discusses how phases $1$ and $2$ of {\em FairRec} work towards better performance on producer and customer sides.
Finally in \cref{subsec:differential_alpha}, we also test {\em FairRec} for scenarios where the platforms may want to ensure different levels of exposure guarantee for different categories of producers.
For evaluating \textit{FairRec} and the baselines, we use the following \textit{producer-side} and \textit{customer-side} metrics.
		\subsubsection{\bf Producer-Side Metrics}
		The evaluation metrics for capturing the fairness and efficiency among the producers are:
		
		\noindent \textbf{\textit{Fraction of Satisfied Producers ($H$):}}
		We call a producer satisfied iff its exposure is more than the minimum exposure guarantee $\overline{E}$. The fraction of \textit{satisfied} producers can be calculated as below.
		\begin{equation}\small
			H=\frac{1}{|P|}\sum_{p \in P}\mathbbm{1}_{E_p\geq \overline{E}}\label{metric_H}
		\end{equation}
		{$\mathbbm{1}_{E_p\geq \overline{E}}$ is $1$ if $E_p\geq \overline{E}$, otherwise $0$.} The value of $H$ ranges between $0$ and $1$. The higher the $H$, the fairer is the recommender system to producers.
		
		\noindent \textbf{\textit{Inequality in Producer Exposures ($Z$):}}
		We earlier observed in \cref{motivation} that conventional top-$k$ recommendation causes huge disparity in individual producer exposures.
		To capture how unequal the individual producer exposures are, we employ an entropy-like measure as below.
		\begin{equation}\small
			Z=-\sum_{p\in P}\Big(\frac{E_p}{m\times k}\Big)\cdot\log_{n}\Big(\frac{E_p}{m\times k}\Big)\label{metric_Z}
		\end{equation}
		Note that the base of the logarithm above is $n$ which is the number of producers. 
		Since each of the $m$ customers is given $k$-sized recommendations, total available exposure is $mk$.
		If every producer gets same exposure, i.e., recommended exactly $\frac{mk}{n}$ times, then the above entropy metric will be: {\small $-n\times\frac{mk/n}{mk}\times \log_n\frac{mk/n}{mk} = -n\times\frac{1}{n}\times \log_n\frac{1}{n} = -n\times\frac{1}{n}\times(-1)=1$}. 
		On the other hand, if only one producer is allowed to get all the exposure, then the value entropy expression in equation-12 will be {\small $-\frac{mk}{mk}\times\log_n\frac{mk}{mk}=0$.
		Thus the range of $Z$ is $[0,1]$}.
		The lower the $Z$, the more unequal individual producer exposures are.
		
		\noindent \textbf{\textit{Exposure Loss on Producers ($L$):}} 
		As {\em FairRec} tries to ensure minimum exposure guarantee for all the producers,
		some producers may receive a lower exposure in comparison to what they would have got in top-$k$ recommendations. To capture this, we compute the loss $L$ as the mean amount of impact (loss in exposure) caused by {\em FairRec}, compared to the top-$k$ recommendations.
		\begin{equation}\small
			L^\text{<method>}=\frac{1}{n}\sum_{p \in P}\text{max}\Bigg(\frac{\big(E^\text{top-$k$}_p-E^\text{<method>}_p\big)}{E^\text{top-$k$}_p},0\Bigg)\label{metric_L}
		\end{equation}
		This metric takes the top-k recommendations (which is the regular recommendation based on the estimated relevance scores, but with no additional constraints) as a reference point, and then evaluates how much exposure is lost on average if some other method is used for the recommendations.
		The lower the exposure loss metric, lower is the negative impact, and the better is the recommendation algorithm.
		
		\subsubsection{\bf Customer-Side Metrics}\label{sec:cust_metric}
		The evaluation metrics for capturing the fairness and efficiency among the customers are:
		
		\noindent \textbf{\textit{Mean Average Envy ($Y$):}}
		Although {\em FairRec} ensures {\em EF1} guarantee for customers by design, here we capture how effectively this guarantee can reduce overall envy among customers in comparison to the baselines.
		We define the mean average envy as below.
		\begin{equation}\small
			Y=\dfrac{1}{n}\sum_{u\in U}\dfrac{1}{n-1}\sum_{\substack{u'\in U\\ u'\ne u}}\text{envy}(u,u')\label{metric_Y}
		\end{equation}
		where {\small $\text{envy}(u,u')=\text{max}\Big(\big(\phi_u(R_{u'})-\phi_u(R_u)\big),0\Big)$} denoting how much $u$ envies $u'$, which is the extra utility $u$ would have received if she had received the recommendation that had been given to $u'$ ($R_{u'}$) instead of her own allocated recommendation $R_u$. The lower the envy ($Y$), the fairer the recommender system is for the customers.
		
		\noindent \textbf{\textit{Mean and Standard Deviation of Customer Utilities (using {\small $\mu_\phi$, $std_\phi$}):}}
		{\em FairRec} may not allocate the most relevant products to the customers,
		which may introduce a loss in customer utilities. This loss can be captured using the expression {\small $\mu_\phi=\frac{1}{m}\sum_{u\in U}\phi_{u}(R_{u})$}. 
		Higher the utility (i.e., lower utility loss), the more efficient is the recommender system for the customers. 
		We also calculate the standard deviation of customer utilities, that is, {\small $std_\phi = std_{u\in U} (\phi_u(R_u))$}. The lower the $std$, lesser is the disparity in individual customer utilities.
\begin{figure*}[t!]
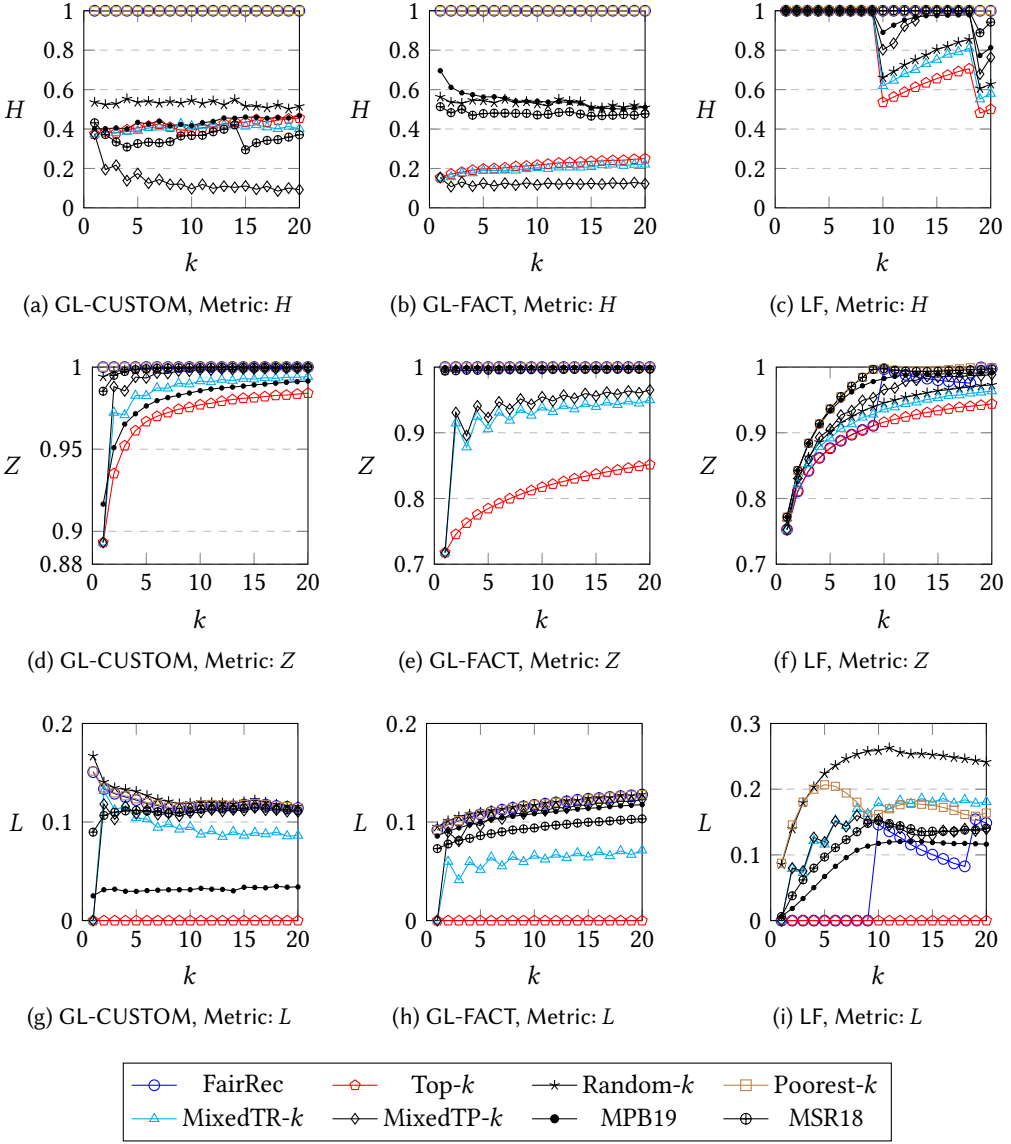

	\center{
		\subfloat[{GL-CUSTOM, Metric: \hyperref[metric_H]{$H$}}]{\pgfplotsset{width=0.32\textwidth,height=0.3\textwidth,compat=1.9}
			\input{figures/gl_1_H}\label{fig:gl_1_H}}
		\hfil
		\subfloat[{GL-FACT, Metric: \hyperref[metric_H]{$H$}}]{\pgfplotsset{width=0.32\textwidth,height=0.3\textwidth,compat=1.9}
			\input{figures/gl_2_H}\label{fig:gl_2_H}}
		\hfil
		\subfloat[{LF, Metric: \hyperref[metric_H]{$H$}}]{\pgfplotsset{width=0.32\textwidth,height=0.3\textwidth,compat=1.9}
			\input{figures/lf_H}\label{fig:lf_H}}
		\hfil
		\subfloat[{GL-CUSTOM, Metric: \hyperref[metric_Z]{$Z$}}]{\pgfplotsset{width=0.32\textwidth,height=0.3\textwidth,compat=1.9}
			\input{figures/gl_1_Z}\label{fig:gl_1_Z}}
		\hfil
		\subfloat[{GL-FACT, Metric: \hyperref[metric_Z]{$Z$}}]{\pgfplotsset{width=0.32\textwidth,height=0.3\textwidth,compat=1.9}
			\input{figures/gl_2_Z}\label{fig:gl_2_Z}}
		\hfil
		\subfloat[{LF, Metric: \hyperref[metric_Z]{$Z$}}]{\pgfplotsset{width=0.32\textwidth,height=0.3\textwidth,compat=1.9}
			\input{figures/lf_Z}\label{fig:lf_Z}}
		\hfil
		\subfloat[{GL-CUSTOM, Metric: \hyperref[metric_L]{$L$}}]{\pgfplotsset{width=0.32\textwidth,height=0.3\textwidth,compat=1.9}
			\input{figures/gl_1_L}\label{fig:gl_1_L}}
		\hfil
		\subfloat[{GL-FACT, Metric: \hyperref[metric_L]{$L$}}]{\pgfplotsset{width=0.32\textwidth,height=0.3\textwidth,compat=1.9}
			\input{figures/gl_2_L}\label{fig:gl_2_L}}
		\hfil
		\subfloat[{LF, Metric: \hyperref[metric_L]{$L$}}]{\pgfplotsset{width=0.32\textwidth,height=0.3\textwidth,compat=1.9}
			\input{figures/lf_L}\label{fig:lf_L}}
		\vfil
		\subfloat{\pgfplotsset{width=.7\textwidth,compat=1.9}
			\begin{tikzpicture}
			\begin{customlegend}[legend entries={{FairRec},{Top-$k$},{Random-$k$},{Poorest-$k$},{MixedTR-$k$},{MixedTP-$k$},{MPB19},{MSR18}},legend columns=4,legend style={/tikz/every even column/.append style={column sep=0.3cm}}]
			\addlegendimage{blue,mark=o,sharp plot}
			\addlegendimage{red,mark=pentagon,sharp plot}
			\addlegendimage{black,mark=star,sharp plot}
			\addlegendimage{brown,mark=square,sharp plot}
			\addlegendimage{cyan,mark=triangle,sharp plot}
			\addlegendimage{black,mark=diamond,sharp plot}
			\addlegendimage{black,mark=*,mark size=1.5pt,sharp plot}
			\addlegendimage{black,mark=oplus,sharp plot}
			\end{customlegend}
			\end{tikzpicture}}
	}\caption{Producer-Side Performances with MMS Guarantee. First row: fraction of satisfied producers (\hyperref[metric_H]{$H$}). Second row: inequality in producer exposures (\hyperref[metric_Z]{$Z$}). Third row: exposure loss on producers (\hyperref[metric_L]{$L$}).}\label{fig:producer_side_mms}
\end{figure*}
\subsection{Experiments with MMS Guarantee}\label{mms_experiments}
\begin{figure*}[t!]
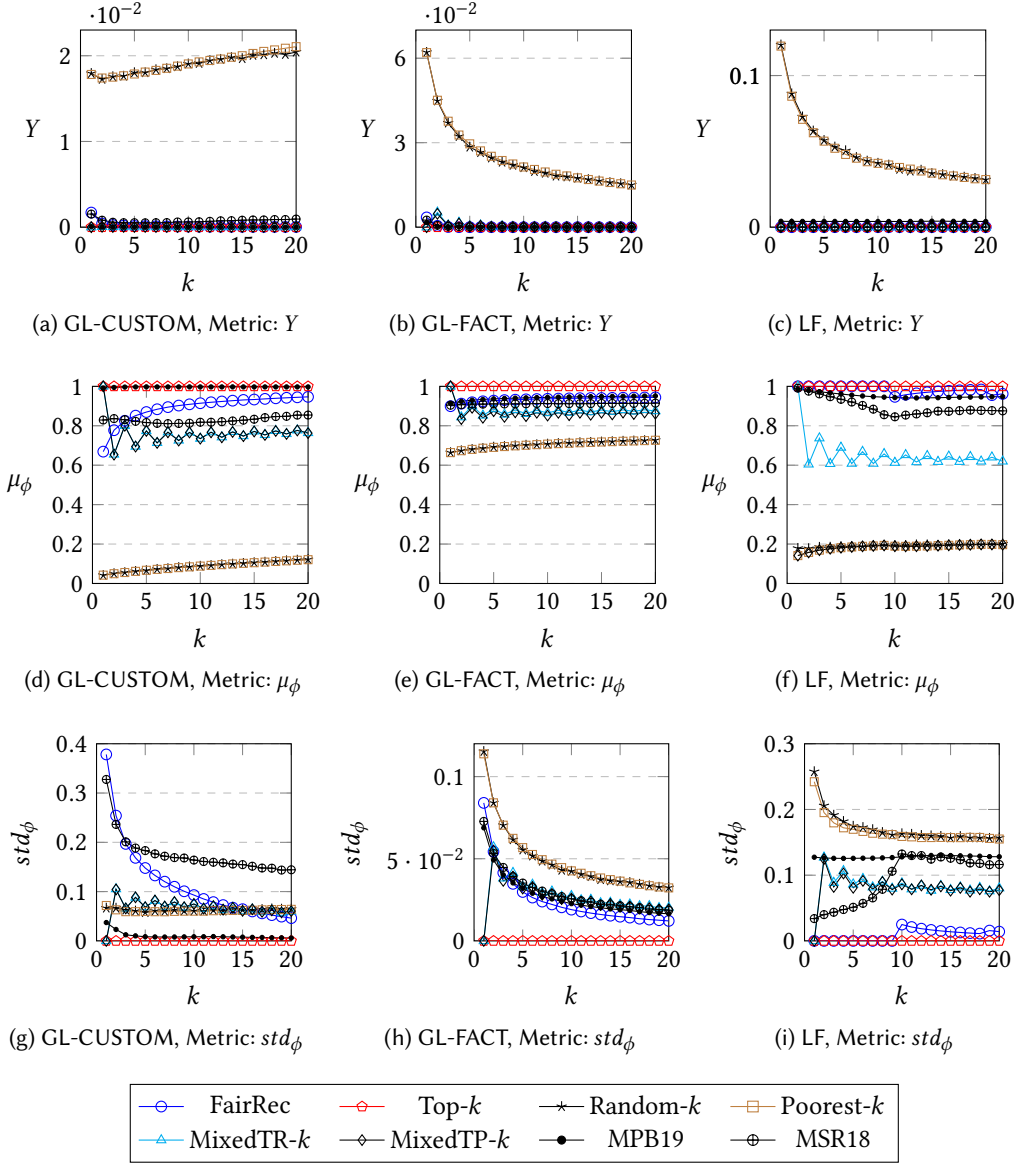

	\center{
		\subfloat[{GL-CUSTOM, Metric: \hyperref[metric_Y]{$Y$}}]{\pgfplotsset{width=0.32\textwidth,height=0.3\textwidth,compat=1.9}
			\input{figures/gl_1_Y}\label{fig:gl_1_Y}}
		\hfil
		\subfloat[{GL-FACT, Metric: \hyperref[metric_Y]{$Y$}}]{\pgfplotsset{width=0.32\textwidth,height=0.3\textwidth,compat=1.9}
			\input{figures/gl_2_Y}\label{fig:gl_2_Y}}
		\hfil
		\subfloat[{LF, Metric: \hyperref[metric_Y]{$Y$}}]{\pgfplotsset{width=0.32\textwidth,height=0.3\textwidth,compat=1.9}
			\input{figures/lf_Y}\label{fig:lf_Y}}
		\vfil
		\subfloat[{GL-CUSTOM, Metric: \hyperref[sec:cust_metric]{${\mu_{\phi}}$}}]{\pgfplotsset{width=0.32\textwidth,height=0.3\textwidth,compat=1.9}
			\input{figures/gl_1_mu}\label{fig:gl_1_mu}}
		\hfil
		\subfloat[{GL-FACT, Metric: \hyperref[sec:cust_metric]{${\mu_{\phi}}$}}]{\pgfplotsset{width=0.32\textwidth,height=0.3\textwidth,compat=1.9}
			\input{figures/gl_2_mu}\label{fig:gl_2_mu}}
		\hfil
		\subfloat[{LF, Metric: \hyperref[sec:cust_metric]{${\mu_{\phi}}$}}]{\pgfplotsset{width=0.32\textwidth,height=0.3\textwidth,compat=1.9}
			\input{figures/lf_mu}\label{fig:lf_mu}}
		\vfil
		\subfloat[{GL-CUSTOM, Metric: \hyperref[sec:cust_metric]{${std_\phi}$}}]{\pgfplotsset{width=0.3\textwidth,height=0.3\textwidth,compat=1.9}
			\input{figures/gl_1_sigma}\label{fig:gl_1_sigma}}
		\hfil
		\subfloat[{GL-FACT, Metric: \hyperref[sec:cust_metric]{${std_\phi}$}}]{\pgfplotsset{width=0.3\textwidth,height=0.3\textwidth,compat=1.9}
			\input{figures/gl_2_sigma}\label{fig:gl_2_sigma}}
		\hfil
		\subfloat[{LF, Metric: \hyperref[sec:cust_metric]{${std_\phi}$}}]{\pgfplotsset{width=0.3\textwidth,height=0.3\textwidth,compat=1.9}
			\input{figures/lf_sigma}\label{fig:lf_sigma}}
		\vfil
		\subfloat{\pgfplotsset{width=.7\textwidth,compat=1.9}
			\begin{tikzpicture}
			\begin{customlegend}[legend entries={{FairRec},{Top-$k$},{Random-$k$},{Poorest-$k$},{MixedTR-$k$},{MixedTP-$k$},{MPB19},{MSR18}},legend columns=4,legend style={/tikz/every even column/.append style={column sep=0.3cm}}]
			\addlegendimage{blue,mark=o,sharp plot}
			\addlegendimage{red,mark=pentagon,sharp plot}
			\addlegendimage{black,mark=star,sharp plot}
			\addlegendimage{brown,mark=square,sharp plot}
			\addlegendimage{cyan,mark=triangle,sharp plot}
			\addlegendimage{black,mark=diamond,sharp plot}
			\addlegendimage{black,mark=*,mark size=1.5pt,sharp plot}
			\addlegendimage{black,mark=oplus,sharp plot}
			\end{customlegend}
			\end{tikzpicture}}
	}\caption{Customer-Side Performances with MMS Guarantee. First row: mean average envy (\hyperref[metric_Y]{$Y$}). Second row: mean customer utility (\hyperref[sec:cust_metric]{$\mu_\phi$}). Third row: standard deviation of customer utilities (\hyperref[sec:cust_metric]{$std_\phi$}).}\label{fig:customer_side}
\end{figure*}
		Here we test {\em FairRec} with exposure guarantee $\overline{E}=\left\lfloor\frac{mk}{n}\right\rfloor=$MMS (or $\alpha=1$), recommendation size $k$ in $1$ to $20$, and discuss the results. 
		\subsubsection{\bf Producer-Side Results}\label{producer_side}
		All producer-side results are plotted in Figure-\ref{fig:producer_side_mms}.
		~\\{\bf Producer Satisfaction (\hyperref[metric_H]{$H$}):}
		Figures \ref{fig:gl_1_H}, \ref{fig:gl_2_H}, and \ref{fig:lf_H} show that both {\em FairRec} and poorest-$k$ perform the best while top-$k$ performs badly; this is because both {\em FairRec} and poorest-$k$ explicitly try to ensure larger exposure for producers while top-$k$ considers only the preferences of the customers.
		Similar to mixedTR-$k$ and mixedTP-$k$, MPB19 also fails to ensure MMS exposure as it does not explicitly ensure an exposure guarantee.
		On the other hand, MSR18 has exposure-based constraint for producers similar to the exposure guarantee in FairRec.
		However, MSR18 fails to exactly satisfy exposure constraint as it follows an approximate approach of constraint optimization while we find that the producer exposures in MSR18 results are quite close to the MMS exposure even though not crossing it.
		Thus, even though MSR18 performs quite bad in $H$, we find it to be performing almost as good as FairRec in the other two producer-side metrics.
		~\\{\bf Exposure Inequality (\hyperref[metric_Z]{$Z$}):} 
		Figures \ref{fig:gl_1_Z}, \ref{fig:gl_2_Z}, and \ref{fig:lf_Z} show that poorest-$k$ has the lowest inequality in exposure while {\em FairRec} and random-$k$ perform similar or slightly less than that;
		on the other hand top-$k$ performs the worst as it is highly customer-centric.
		Since poorest-$k$ strategy makes the deterministic selection of less exposed producers thereby making smarter choices than random-$k$, poorest-$k$ ensures less inequality (i.e., high $Z$) than random-$k$.
		Performances of both mixedTR-$k$ and mixedTP-$k$ lie in between the best and the worst lines as they mix customer-centric top half and producer-centric bottom half recommendations.
		MPB19 shows inconsistent performances across the datasets.
		On the other hand, MSR18 shows low exposure inequality (high $Z$) very close to that of FairRec and poorest-$k$, since MSR18 is observed to be ensuring exposures close to MMS for most of the producers.
		In summary, FairRec, poorest-$k$, and MSR18 seem to be good at maintaining fairness on the producer-side by keeping exposure inequality low.
		~\\{\bf Exposure Loss (\hyperref[metric_L]{$L$}):}
		Figures \ref{fig:gl_1_L}, \ref{fig:gl_2_L}, and \ref{fig:lf_L} show that random-$k$ and poorest-$k$ cause the highest amounts of exposure loss in comparison to top-$k$;
		this is because both of them favor equality in producer exposure (random-$k$ gives equal chance to all producers to be recommended while poorest-$k$ tries to increase the exposure of least exposed producer).
		On the other hand, mixedTR-$k$, mixedTP-$k$ cause smaller losses as only up to half of their recommendations are different from top-$k$.
		{\em FairRec} causes only up to $0.2$ fraction or $20$\% loss in exposure in comparison to top-$k$ owing to its intelligent selection approach.
		
		It is worth noticing that MMS for {\em LF} is low (MMS$=0$ for $k<10$, MMS$=1$ for $k\in[10,18]$, MMS$=2$ for $k\in[20,29]$,...).
		MMS is satisfied for all producers until $k=9$;
		but at k=10, MMS is not guaranteed for all producers, and thus, we see a drop in performance at $k=10$ which happens again at $k=19$. 
		Such changes in MMS specific to {\em LF} make its plots different from other datasets.
		In summary, both {\em FairRec}, poorest-$k$ perform the best in producer fairness while they cause exposure loss for very popular producers to compensate for the exposure given to less popular producers.
		\subsubsection{\bf Customer-Side Results}\label{customer_side}
		All customer-side results are plotted in Figure-\ref{fig:customer_side}.
		~\\{\bf Mean Average Envy (\hyperref[metric_Y]{$Y$}):}
		Figures \ref{fig:gl_1_Y}, \ref{fig:gl_2_Y}, and \ref{fig:lf_Y} reveal that top-$k$ causes the lowest possible (i.e., $0$) mean average envy among the customers;
		this is because it gives the maximum possible utility of $1$ to every customer thereby leaving no chances of envy among customers.
		Even mixedTR-$k$, mixedTP-$k$, MPB19, and MSR18 also show similarly low envy.
		{\em FairRec} generates very low values of envy which are very comparable to those of top-$k$ here.
		On the other hand both random-$k$ and poorest-$k$ cause the highest envy as they do not consider customer preferences at all during recommendation.
		~\\{\bf Mean and Standard Deviation of Customer Utility (\hyperref[sec:cust_metric]{$\mu_\phi$},\hyperref[sec:cust_metric]{$std_\phi$}):}
		From Figures \ref{fig:gl_1_mu}, \ref{fig:gl_2_mu}, and \ref{fig:lf_mu}, we see that both random-$k$ and poorest-$k$ cause huge losses in customer utility (i.e., low customer utility) as they neglect customer preferences.
		The mixedTR-$k$ and mixedTP-$k$ perform moderately.
		On the other hand, {\em FairRec} causes very small utility loss and performs almost at par with the customer-centric top-$k$.
		The utility losses in MPB19 and MSR18 are higher as they do not guarantee anything on the customer-side.
		This certifies that {\em FairRec} strikes a good balance between customer utility and producer fairness.  
		The standard deviation plots: figures-\ref{fig:gl_1_sigma}, \ref{fig:gl_2_sigma}, and \ref{fig:lf_sigma} reveal that for larger sizes of recommendation, random-$k$ and poorest-$k$ show large disparities in customer utilities while {\em FairRec}, mixedTR-$k$ and mixedTP-$k$ show relatively fewer disparities.
		As top-$k$ is customer-centric and provides the maximum utility of $1$ to all the customers, it shows $0$ standard deviation.
		Besides MPB19 and MSR18 show higher disparities on customer-side as they do not specifically guarantee anything on customer-side.
		
		~\\In summary, FairRec strikes a good balance between fairness on both producer-side and customer-side while causing only marginal losses in customer utility.
\subsection{Experiments with $\alpha$-MMS Guarantee}\label{alpha_mms_experiments}
Here we fix $k=20$, and test {\em FairRec} with different values of minimum exposure guarantee i.e., $\overline{E}=\left\lfloor\alpha\cdot \frac{mk}{n}\right\rfloor$ (where $0\leq\!\alpha\!\leq1$) by varying $\alpha$ in between $0$ and $1$ (or in other words varying $\overline{E}$ in between $0$ and MMS);
we plot the results in Figures \ref{fig:alpha_prod_performances} and \ref{fig:alpha_cust_performances}.

\subsubsection{\bf Producer-Side Results}
All the relevant producer-side results are plotted in Figure \ref{fig:alpha_prod_performances}.
~\\{\bf Producer Satisfaction (\hyperref[metric_H]{$H$}):}
Figures \ref{fig:gl_1_alpha_H}, \ref{fig:gl_2_alpha_H}, and \ref{fig:lf_alpha_H} reveal that {\em FairRec} satisfies almost all the producers (as \hyperref[metric_H]{$H$} is close to $1$) for all the tested $\alpha$ settings in all the datasets.
However the performances of the baselines in terms of metric \hyperref[metric_H]{$H$} reduces with the increase in $\alpha$;
this is because the criteria for the producer satisfaction (as $\overline{E}\propto \alpha$) increases with the increase in $\alpha$ while the baseline results do not explicitly change with the change in $\alpha$.
~\\{\bf Exposure Inequality (\hyperref[metric_Z]{$Z$}):}
From Figures \ref{fig:gl_1_alpha_Z}, \ref{fig:gl_2_alpha_Z}, and \ref{fig:lf_alpha_Z}, we see that increasing minimum exposure guarantee (i.e., increasing $\alpha$) results in lower inequality in producer exposures (as high \hyperref[metric_Z]{$Z$} signifies lower inequality) for {\em FairRec} in all the cases.
On the other hand, as the baselines do not depend on $\alpha$, their performances remain the same;
thus the baseline performances are just horizontal straight lines.
At $\alpha=0$, as the exposure guarantee $\overline{E}$ by {\em FairRec} becomes $0$, the results given by {\em FairRec} are the same as that of top-$k$ in all the datasets.
While {\em FairRec}'s performance in terms of \hyperref[metric_Z]{$Z$} crosses that of mixedTR-$k$ at $\alpha=0.7$, $\alpha=0.7$, and $\alpha=0.5$ in GL-CUSTOM, GL-FACT, and LF respectively, it becomes close to that of producer-centric poorest-$k$ at $\alpha=1$ in all the datasets.
~\\{\bf Exposure Loss (\hyperref[metric_L]{$L$}):}
Increasing values of \hyperref[metric_L]{$L$} are observed for increased $\alpha$ in Figures-\ref{fig:gl_1_alpha_L}, \ref{fig:gl_2_alpha_L}, \ref{fig:lf_alpha_L}.
This suggests that increasing minimum exposure guarantee can cause higher exposure losses for previously popular producers.
Just like \hyperref[metric_Z]{$Z$}, here also the baseline performances are just horizontal straight lines as they do not depend on $\alpha$.
We see that at $\alpha=0$, {\em FairRec} performs same as top-$k$ with no losses.
With the increase in $\alpha$ settings, the losses increase and at $\alpha=1$, the losses in {\em FairRec} are very close to those of poorest-$k$ (as discussed earlier, poorest-$k$ is the best performing baselines for producer side).
\begin{figure*}[t!]
	\center{
		\subfloat[{GL-CUSTOM, Metric: \hyperref[metric_H]{$H$}}]{\pgfplotsset{width=0.3\textwidth,height=0.23\textwidth,compat=1.9}
			\input{figures/gl_1_alpha_H}\label{fig:gl_1_alpha_H}}
		\hfil
		\subfloat[{GL-FACT, Metric: \hyperref[metric_H]{$H$}}]{\pgfplotsset{width=0.3\textwidth,height=0.23\textwidth,compat=1.9}
			\input{figures/gl_2_alpha_H}\label{fig:gl_2_alpha_H}}
		\hfil
		\subfloat[{LF, Metric: \hyperref[metric_H]{$H$}}]{\pgfplotsset{width=0.3\textwidth,height=0.23\textwidth,compat=1.9}
			\input{figures/lf_alpha_H}\label{fig:lf_alpha_H}}
		\vfil
		\subfloat[{GL-CUSTOM, Metric: \hyperref[metric_Z]{$Z$}}]{\pgfplotsset{width=0.3\textwidth,height=0.23\textwidth,compat=1.9}
			\input{figures/gl_1_alpha_Z}\label{fig:gl_1_alpha_Z}}
		\hfil
		\subfloat[{GL-FACT, Metric: \hyperref[metric_Z]{$Z$}}]{\pgfplotsset{width=0.3\textwidth,height=0.23\textwidth,compat=1.9}
			\input{figures/gl_2_alpha_Z}\label{fig:gl_2_alpha_Z}}
		\hfil
		\subfloat[{LF, Metric: \hyperref[metric_Z]{$Z$}}]{\pgfplotsset{width=0.3\textwidth,height=0.23\textwidth,compat=1.9}
			\input{figures/lf_alpha_Z}\label{fig:lf_alpha_Z}}
		\vfil
		\subfloat[{GL-CUSTOM, Metric: \hyperref[metric_L]{$L$}}]{\pgfplotsset{width=0.3\textwidth,height=0.23\textwidth,compat=1.9}
			\input{figures/gl_1_alpha_L}\label{fig:gl_1_alpha_L}}
		\hfil
		\subfloat[{GL-FACT, Metric: \hyperref[metric_L]{$L$}}]{\pgfplotsset{width=0.3\textwidth,height=0.23\textwidth,compat=1.9}
			\input{figures/gl_2_alpha_L}\label{fig:gl_2_alpha_L}}
		\hfil
		\subfloat[{LF, Metric: \hyperref[metric_L]{$L$}}]{\pgfplotsset{width=0.3\textwidth,height=0.23\textwidth,compat=1.9}
			\input{figures/lf_alpha_L}\label{fig:lf_alpha_L}}
		\vfil
		\subfloat{\pgfplotsset{width=.7\textwidth,compat=1.9}
			\begin{tikzpicture}
			\begin{customlegend}[legend entries={{FairRec},{Top-$k$},{Random-$k$},{Poorest-$k$},{MixedTR-$k$},{MixedTP-$k$}},legend columns=6,legend style={/tikz/every even column/.append style={column sep=0.2cm}}]
			\addlegendimage{blue,mark=o,sharp plot}
			\addlegendimage{red,mark=pentagon,sharp plot}
			\addlegendimage{black,mark=star,sharp plot}
			\addlegendimage{brown,mark=square,sharp plot}
			\addlegendimage{cyan,mark=triangle,sharp plot}
			\addlegendimage{black,mark=diamond,sharp plot}
			\end{customlegend}
			\end{tikzpicture}}
		
	}\caption{Producer-side performances with $\overline{E}=\alpha$MMS guarantee for $k=20$. First row: fraction of satisfied producers (\hyperref[metric_H]{$H$}). Second row: inequality in producer exposures (\hyperref[metric_Z]{$Z$}). Third row: exposure loss on producers (\hyperref[metric_L]{$L$}).}\label{fig:alpha_prod_performances}
	\vspace{-2mm}
\end{figure*}
\subsubsection{\bf Customer-Side Results}
All the relevant customer-side results are plotted in Figure \ref{fig:alpha_cust_performances}. As the baselines do not depend on $\alpha$, the customer-side results of baselines are just horizontal straight lines.
We find almost no change in Mean Average Envy (\hyperref[metric_Y]{$Y$}) of {\em FairRec} with the change in $\alpha$ (refer Figures \ref{fig:gl_1_alpha_Y}, \ref{fig:gl_2_alpha_Y}, \ref{fig:lf_alpha_Y}).
On the other hand, with the increase in $\alpha$ (i.e., higher exposure guarantee for producers) there is a small decrease in customer utility (refer Figures-\ref{fig:gl_1_alpha_mu}, \ref{fig:gl_2_alpha_mu}, \ref{fig:lf_alpha_mu}), and a small increase in the standard deviation of customer utilities.
\begin{figure*}[t!]
	\center{
		\subfloat[{GL-CUSTOM, Metric: \hyperref[metric_Y]{$Y$}}]{\pgfplotsset{width=0.31\textwidth,height=0.23\textwidth,compat=1.9}
			\input{figures/gl_1_alpha_Y}\label{fig:gl_1_alpha_Y}}
		\hfil
		\subfloat[{GL-FACT, Metric: \hyperref[metric_Y]{$Y$}}]{\pgfplotsset{width=0.31\textwidth,height=0.23\textwidth,compat=1.9}
			\input{figures/gl_2_alpha_Y}\label{fig:gl_2_alpha_Y}}
		\hfil
		\subfloat[{LF, Metric: \hyperref[metric_Y]{$Y$}}]{\pgfplotsset{width=0.31\textwidth,height=0.23\textwidth,compat=1.9}
			\input{figures/lf_alpha_Y}\label{fig:lf_alpha_Y}}
		\vfil
		\subfloat[{GL-CUSTOM, Metric: \hyperref[sec:cust_metric]{$\mu_{\phi}$}}]{\pgfplotsset{width=0.3\textwidth,height=0.23\textwidth,compat=1.9}
			\input{figures/gl_1_alpha_mu}\label{fig:gl_1_alpha_mu}}
		\hfil
		\subfloat[{GL-FACT, Metric: \hyperref[sec:cust_metric]{$\mu_{\phi}$}}]{\pgfplotsset{width=0.3\textwidth,height=0.23\textwidth,compat=1.9}
			\input{figures/gl_2_alpha_mu}\label{fig:gl_2_alpha_mu}}
		\hfil
		\subfloat[{LF, Metric: \hyperref[sec:cust_metric]{$\mu_{\phi}$}}]{\pgfplotsset{width=0.3\textwidth,height=0.23\textwidth,compat=1.9}
			\input{figures/lf_alpha_mu}\label{fig:lf_alpha_mu}}
		\vfil		
		\subfloat[{GL-CUSTOM, Metric: \hyperref[sec:cust_metric]{$std_{\phi}$}}]{\pgfplotsset{width=0.29\textwidth,height=0.23\textwidth,compat=1.9}
			\input{figures/gl_1_alpha_sigma}\label{fig:gl_1_alpha_sigma}}
		\hfil
		\subfloat[{GL-FACT, Metric: \hyperref[sec:cust_metric]{$std_{\phi}$}}]{\pgfplotsset{width=0.29\textwidth,height=0.23\textwidth,compat=1.9}
			\input{figures/gl_2_alpha_sigma}\label{fig:gl_2_alpha_sigma}}
		\hfil
		\subfloat[{LF, Metric: \hyperref[sec:cust_metric]{$std_{\phi}$}}]{\pgfplotsset{width=0.29\textwidth,height=0.23\textwidth,compat=1.9}
			\input{figures/lf_alpha_sigma}\label{fig:lf_alpha_sigma}}
		\vfil
		\subfloat{\pgfplotsset{width=.7\textwidth,compat=1.9}
			\begin{tikzpicture}
			\begin{customlegend}[legend entries={{FairRec},{Top-$k$},{Random-$k$},{Poorest-$k$},{MixedTR-$k$},{MixedTP-$k$}},legend columns=6,legend style={/tikz/every even column/.append style={column sep=0.2cm}}]
			\addlegendimage{blue,mark=o,sharp plot}
			\addlegendimage{red,mark=pentagon,sharp plot}
			\addlegendimage{black,mark=star,sharp plot}
			\addlegendimage{brown,mark=square,sharp plot}
			\addlegendimage{cyan,mark=triangle,sharp plot}
			\addlegendimage{black,mark=diamond,sharp plot}
			\end{customlegend}
			\end{tikzpicture}}
	}\caption{Customer-side performances with $\overline{E}=\alpha$MMS guarantee for $k=20$. First row: mean average envy (\hyperref[metric_Y]{$Y$}). Second row: mean customer utility (\hyperref[sec:cust_metric]{$\mu_\phi$}). Third row: standard deviation of customer utilities (\hyperref[sec:cust_metric]{$std_\phi$}).}\label{fig:alpha_cust_performances}
	\vspace{-2mm}
\end{figure*}

\noindent In summary, although a larger exposure guarantee can help platforms achieve better producer fairness, it might hurt the overall customer satisfaction and also the satisfaction of highly popular producers of the platforms.
Thus, the platforms, which are interested in similar minimum exposure guarantees, should not ignore the above trade-offs.

\subsection{Interplay Between Phase-1 and Phase-2 of FairRec}\label{subsec:phase_interplay}
To understand how phase-1 and phase-2 of {\em FairRec} work towards better performance on producer and customer sides, we plot the metrics at the end of both phase-1 and phase-2 (i.e., the end of {\em FairRec}) in Figures \ref{fig:phase_interplay_producer_side} and \ref{fig:phase_interplay_customer_side}.
\begin{figure*}[t!]
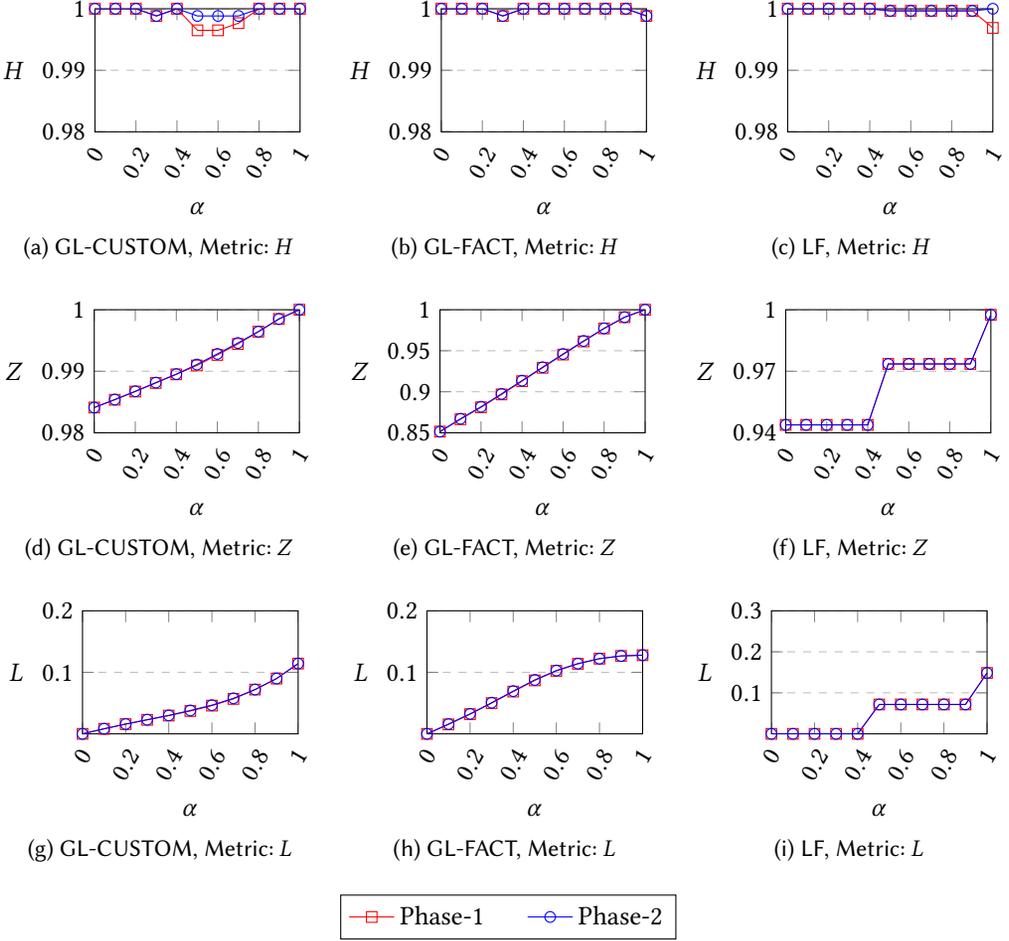

	\center{
		\subfloat[{GL-CUSTOM, Metric: \hyperref[metric_H]{$H$}}]{\pgfplotsset{width=0.31\textwidth,height=0.23\textwidth,compat=1.9}
			\input{figures/gl_1_H_phases}\label{fig:gl_1_H_phases}}
		\hfil
		\subfloat[{GL-FACT, Metric: \hyperref[metric_H]{$H$}}]{\pgfplotsset{width=0.31\textwidth,height=0.23\textwidth,compat=1.9}
			\input{figures/gl_2_H_phases}\label{fig:gl_2_H_phases}}
		\hfil
		\subfloat[{LF, Metric: \hyperref[metric_H]{$H$}}]{\pgfplotsset{width=0.31\textwidth,height=0.23\textwidth,compat=1.9}
			\input{figures/lf_H_phases}\label{fig:lf_H_phases}}
		\hfil
		\subfloat[{GL-CUSTOM, Metric: \hyperref[metric_Z]{$Z$}}]{\pgfplotsset{width=0.31\textwidth,height=0.23\textwidth,compat=1.9}
			\input{figures/gl_1_Z_phases}\label{fig:gl_1_Z_phases}}
		\hfil
		\subfloat[{GL-FACT, Metric: \hyperref[metric_Z]{$Z$}}]{\pgfplotsset{width=0.31\textwidth,height=0.23\textwidth,compat=1.9}
			\input{figures/gl_2_Z_phases}\label{fig:gl_2_Z_phases}}
		\hfil
		\subfloat[{LF, Metric: \hyperref[metric_Z]{$Z$}}]{\pgfplotsset{width=0.31\textwidth,height=0.23\textwidth,compat=1.9}
			\input{figures/lf_Z_phases}\label{fig:lf_Z_phases}}
		\hfil
		\subfloat[{GL-CUSTOM, Metric: \hyperref[metric_L]{$L$}}]{\pgfplotsset{width=0.32\textwidth,height=0.23\textwidth,compat=1.9}
			\input{figures/gl_1_L_phases}\label{fig:gl_1_L_phases}}
		\hfil
		\subfloat[{GL-FACT, Metric: \hyperref[metric_L]{$L$}}]{\pgfplotsset{width=0.32\textwidth,height=0.23\textwidth,compat=1.9}
			\input{figures/gl_2_L_phases}\label{fig:gl_2_L_phases}}
		\hfil
		\subfloat[{LF, Metric: \hyperref[metric_L]{$L$}}]{\pgfplotsset{width=0.32\textwidth,height=0.23\textwidth,compat=1.9}
			\input{figures/lf_L_phases}\label{fig:lf_L_phases}}
		\vfil
		\subfloat{\pgfplotsset{width=.5\textwidth,compat=1.9}
			\begin{tikzpicture}
			\begin{customlegend}[legend entries={{Phase-$1$},{Phase-$2$}},legend columns=5,legend style={/tikz/every even column/.append style={column sep=0.5cm}}]			
			\addlegendimage{red,mark=square,sharp plot}
			\addlegendimage{blue,mark=o,sharp plot}
			\end{customlegend}
			\end{tikzpicture}}
	}\caption{Phase interplay figures on producer-side for $k=20$. First row: fraction of satisfied producers (\hyperref[metric_H]{$H$}). Second row: inequality in producer exposures (\hyperref[metric_Z]{$Z$}). Third row: exposure loss on producers (\hyperref[metric_L]{$L$}).}\label{fig:phase_interplay_producer_side}
\end{figure*}
\subsubsection{\bf Producer-Side Results}
All the relevant producer-side results are plotted in Figure \ref{fig:phase_interplay_producer_side}.
We find that the metrics evaluated after phase-1 and phase-2 of {\em FairRec} have almost the same values in all the datasets;
this is because the task of reducing exposure inequality through minimum exposure guarantee to producers happens in phase-1 of {\em FairRec}, and phase-2 of {\em FairRec} does not explicitly work towards these goals.
Thus, there is very little change in the relevant metrics of producer satisfaction (\hyperref[metric_H]{$H$} in Figures \ref{fig:gl_1_H_phases}, \ref{fig:gl_2_H_phases}, \ref{fig:lf_H_phases}), exposure inequality (\hyperref[metric_Z]{$Z$} in Figures \ref{fig:gl_1_Z_phases}, \ref{fig:gl_2_Z_phases}, \ref{fig:lf_Z_phases}), and exposure loss (\hyperref[metric_L]{$L$} in Figures \ref{fig:gl_1_L_phases}, \ref{fig:gl_2_L_phases}, \ref{fig:lf_Z_phases}) after phase-2 than those at the end of phase-1.
However, at a few settings of $\alpha$ in GL-CUSTOM and LF where phase-1 falls a bit short in producer satisfaction (\hyperref[metric_H]{$H$}), phase-2 seems to improve it by a small extent (refer Figures \ref{fig:gl_1_H_phases} and \ref{fig:lf_H_phases} respectively).
\begin{figure*}[t!]
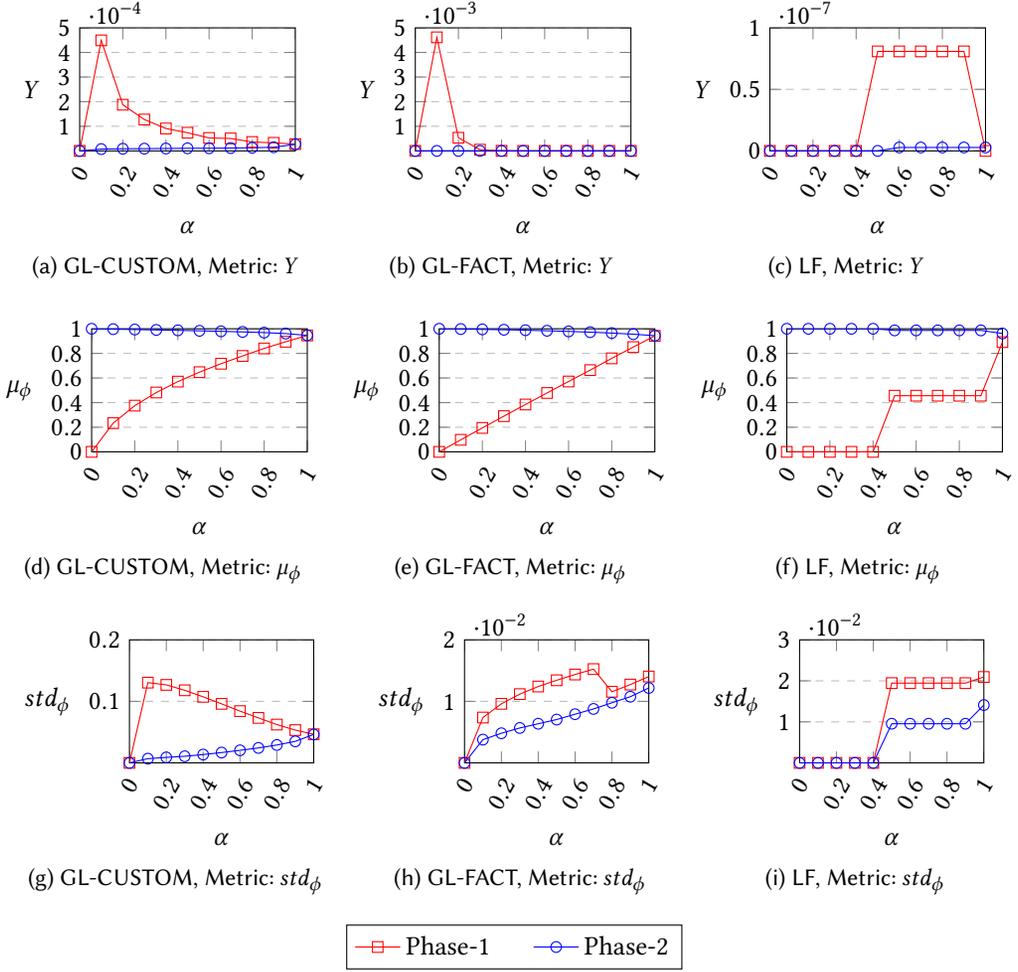

	\center{
		\subfloat[{GL-CUSTOM, Metric: \hyperref[metric_Y]{$Y$}}]{\pgfplotsset{width=0.32\textwidth,height=0.23\textwidth,compat=1.9}
			\input{figures/gl_1_Y_phases}\label{fig:gl_1_Y_phases}}
		\hfil
		\subfloat[{GL-FACT, Metric: \hyperref[metric_Y]{$Y$}}]{\pgfplotsset{width=0.32\textwidth,height=0.23\textwidth,compat=1.9}
			\input{figures/gl_2_Y_phases}\label{fig:gl_2_Y_phases}}
		\hfil
		\subfloat[{LF, Metric: \hyperref[metric_Y]{$Y$}}]{\pgfplotsset{width=0.32\textwidth,height=0.23\textwidth,compat=1.9}
			\input{figures/lf_Y_phases}\label{fig:lf_Y_phases}}
		\hfil
		\subfloat[{GL-CUSTOM, Metric: \hyperref[sec:cust_metric]{${\mu_{\phi}}$}}]{\pgfplotsset{width=0.32\textwidth,height=0.23\textwidth,compat=1.9}
			\input{figures/gl_1_mu_phases}\label{fig:gl_1_mu_phases}}
		\hfil
		\subfloat[{GL-FACT, Metric: \hyperref[sec:cust_metric]{${\mu_{\phi}}$}}]{\pgfplotsset{width=0.32\textwidth,height=0.23\textwidth,compat=1.9}
			\input{figures/gl_2_mu_phases}\label{fig:gl_2_mu_phases}}
		\hfil
		\subfloat[{LF, Metric: \hyperref[sec:cust_metric]{${\mu_{\phi}}$}}]{\pgfplotsset{width=0.32\textwidth,height=0.23\textwidth,compat=1.9}
			\input{figures/lf_mu_phases}\label{fig:lf_mu_phases}}
		\hfil
		\subfloat[{GL-CUSTOM, Metric: \hyperref[sec:cust_metric]{${std_\phi}$}}]{\pgfplotsset{width=0.29\textwidth,height=0.23\textwidth,compat=1.9}
			\input{figures/gl_1_sigma_phases}\label{fig:gl_1_sigma_phases}}
		\hfil
		\subfloat[{GL-FACT, Metric: \hyperref[sec:cust_metric]{${std_\phi}$}}]{\pgfplotsset{width=0.29\textwidth,height=0.23\textwidth,compat=1.9}
			\input{figures/gl_2_sigma_phases}\label{fig:gl_2_sigma_phases}}
		\hfil
		\subfloat[{LF, Metric: \hyperref[sec:cust_metric]{${std_\phi}$}}]{\pgfplotsset{width=0.29\textwidth,height=0.23\textwidth,compat=1.9}
			\input{figures/lf_sigma_phases}\label{fig:lf_sigma_phases}}
		\vfil
		\subfloat{\pgfplotsset{width=.5\textwidth,compat=1.9}
			\begin{tikzpicture}
			\begin{customlegend}[legend entries={{Phase-$1$},{Phase-$2$}},legend columns=5,legend style={/tikz/every even column/.append style={column sep=0.5cm}}]
			\addlegendimage{red,mark=square,sharp plot}
			\addlegendimage{blue,mark=o,sharp plot}
			\end{customlegend}
			\end{tikzpicture}}
	}\caption{Phase interplay figures on customer-side for $k=20$. First row: mean average envy (\hyperref[metric_Y]{$Y$}). Second row: mean customer utility (\hyperref[sec:cust_metric]{$\mu_\phi$}). Third row: standard deviation of customer utilities (\hyperref[sec:cust_metric]{$std_\phi$}).}\label{fig:phase_interplay_customer_side}
\end{figure*}
\subsubsection{\bf Customer-Side Results}
All the relevant customer-side results are plotted in Figure \ref{fig:phase_interplay_customer_side}.
For all the $\alpha$ values which correspond to $\overline{E}=0$ (like $\alpha=0$ in GL-CUSTOM and GL-FACT, $\alpha\in \{0.1,0.2,0.3,0.4\}$ in LF), we see that all the metrics \hyperref[metric_Y]{$Y$}, \hyperref[sec:cust_metric]{$\mu_\phi$}, and \hyperref[sec:cust_metric]{$std_\phi$} correspond to $0$ after phase-1 of {\em FairRec} as no allocation happens 
in phase-1;
all the allocations happen in phase-2 thus making the final results same as that of top-$k$ (\hyperref[metric_Y]{$Y$}$=0$, \hyperref[sec:cust_metric]{$\mu_\phi$}$=1$, and \hyperref[sec:cust_metric]{$std_\phi$}$=0$ after phase-2) in all the cases. 
We observe variations in these metrics right from the point where $\alpha$ corresponds to non-zero $\overline{E}$ (i.e., in our experiments $\alpha=0$ in GL-CUSTOM and GL-FACT, $\alpha=0.5$ in LF);
we describe these observations next.
~\\{\bf Mean Average Envy (\hyperref[metric_Y]{$Y$}):}
For smaller non-zero values of $alpha$, we see higher envy after phase-1 (\hyperref[metric_Y]{$Y$} in Figures \ref{fig:gl_1_Y_phases}, \ref{fig:gl_2_Y_phases}, \ref{fig:lf_Y_phases});
however with the increase in $\alpha$, phase-1 envy decreases;
this must be because with the increase in $\alpha$, the number of rounds of allocation in phase-1 increases, and thereby increasing the chances of canceling out envy to some extent in the Greedy Round Robin process.
On the other hand, phase-2 improves the results by reducing \hyperref[metric_Y]{$Y$};
this is because it allows best possible allocations to every customer by setting higher availability of all the products.
However, with the increase in $\alpha$, the amount of improvement happening in phase-2 reduces as more and more allocations happen in phase-1 thereby reducing the number of rounds in phase-2.
~\\{\bf Mean and Standard Deviation of Customer Utility (\hyperref[sec:cust_metric]{$\mu_\phi$},\hyperref[sec:cust_metric]{$std_\phi$}):}
With the increase in $\alpha$, we observe increase in the mean customer utility of phase-1 allocations (refer Figures \ref{fig:gl_1_mu_phases}, \ref{fig:gl_2_mu_phases}, \ref{fig:lf_mu_phases});
this is because the number of allocations in phase-1 increases with the increase in $\alpha$.
Customer utility after phase-2 is higher than that of phase-1 (refer Figures \ref{fig:gl_1_mu_phases}, \ref{fig:gl_2_mu_phases}, \ref{fig:lf_mu_phases});
however, the utility improvement in phase-2 reduces with the increase in $\alpha$ as more number of allocations happen in phase-1 and the number of  rounds in phase-2 becomes less for higher $\alpha$.
We also observe that phase-2 reduces the disparity in customer utilities than what is observed at the end of phase-1 (refer Figures \ref{fig:gl_1_sigma_phases}, \ref{fig:gl_2_sigma_phases}, \ref{fig:lf_sigma_phases});
this is because phase-2 allows best possible allocations to every customer thereby compensating for customer-side inequalities and losses incurred in phase-1.

\noindent In summary, phase-1 of {\em FairRec} tries to achieve better performance from the producer-side, whereas the phase-2 mostly improves the performance from the customer-side.
In addition to that the customer-side improvement in phase-2 is often more when the $\alpha$ values are smaller as it increases the number of allocations happening in the second phase.

\begin{figure*}[t!]
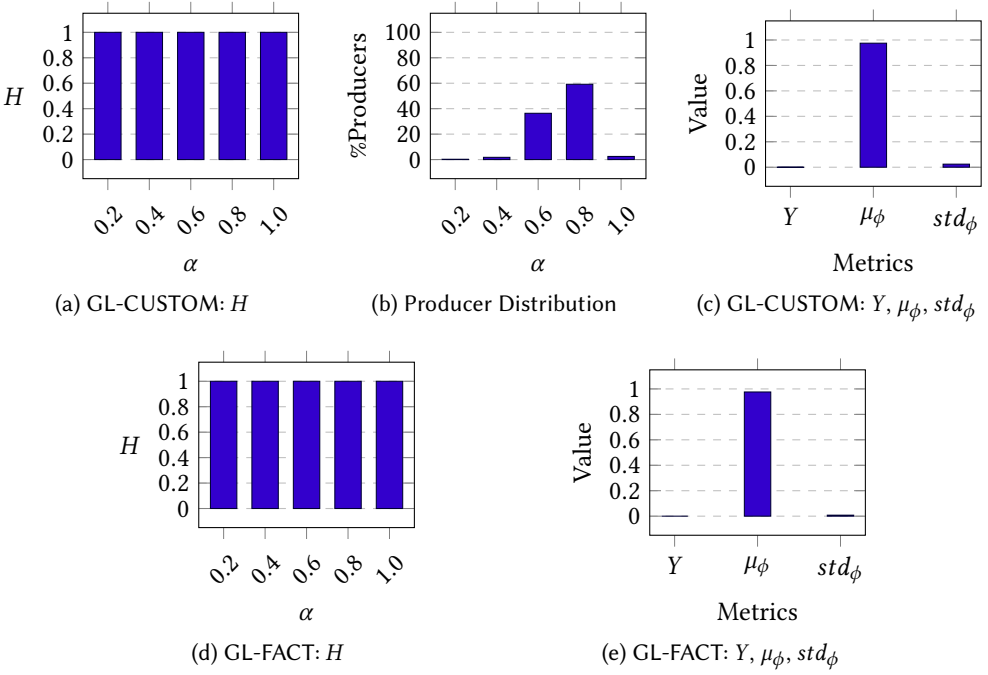

	\center{
		
		\subfloat[GL-CUSTOM: {\hyperref[metric_H]{$H$}}]{\pgfplotsset{width=0.32\textwidth,height=0.27\textwidth,compat=1.9}\input{figures/gl_1_differntial_H}\label{fig:gl_1_differential_H}}
		\hfil
		\subfloat[{Producer Distribution}]{\pgfplotsset{width=0.32\textwidth,height=0.27\textwidth,compat=1.9}
			\input{figures/gl_1_alpha_distribution}\label{fig:gl_1_differential_distribution}}
		\hfil
		\subfloat[GL-CUSTOM: $Y$, $\mu_{\phi}$, $std_{\phi}$]{\pgfplotsset{width=0.32\textwidth,height=0.27\textwidth,compat=1.9}\input{figures/gl_1_differntial_cust}\label{fig:gl_1_differential_cust}}
		\hfil
		\subfloat[GL-FACT: {\hyperref[metric_H]{$H$}}]{\pgfplotsset{width=0.32\textwidth,height=0.27\textwidth,compat=1.9}\input{figures/gl_2_differntial_H}\label{fig:gl_2_differential_H}}
		\hfil
		\subfloat[GL-FACT: $Y$, $\mu_{\phi}$, $std_{\phi}$]{\pgfplotsset{width=0.32\textwidth,height=0.27\textwidth,compat=1.9}\input{figures/gl_2_differntial_cust}\label{fig:gl_2_differential_cust}}
		
	}\caption{Results of FairRec with producer-specific $\alpha$ settings in GL-CUSTOM ($k=20$).}\label{fig:differential_alphas}
\end{figure*}
\subsection{Experiments with Producer-Specific $\alpha$}\label{subsec:differential_alpha}
The platforms may also need to ensure different levels of exposure guarantee for different producers;
for example the platform may want to give more exposure guarantee to high-rated producers than low-rated ones.
Thus, we also test {\em FairRec} in such a scenario.
Here, we fix $\alpha_p= 0.2\times\lfloor rating(p)\rfloor$; i.e., for $rating(p)=5$ the exposure guarantee is $\overline{E}_p=1.0\times MMS$, and similarly for $rating(p)=3.2$ the exposure guarantee is $\overline{E}_p=0.6\times MMS$.
We test {\em FairRec} with the above settings on GL-CUSTOM and GL-FACT, and plot the results in Figure \ref{fig:differential_alphas}.
Figure \ref{fig:gl_1_differential_distribution} shows the distribution of producers in Google Local dataset with different exposure guarantees (set based on their respective ratings).
Figures \ref{fig:gl_1_differential_H} and \ref{fig:gl_2_differential_H} plot $H$ for every group of producers with different exposure guarantees in GL-CUSTOM and GL-FACT respectively;
we find {\em FairRec} to be satisfying all the producers in both the cases.
Metrics $Z$ and $L$ are skipped here, as they are irrelevant in case of different exposure guarantees for different producers.
We plot customer-side results in Figures \ref{fig:gl_1_differential_cust} and \ref{fig:gl_2_differential_cust}, GL-CUSTOM and GL-FACT respectively;
we find that {\em FairRec} is able to achieve high customer utility ($\mu_\phi$ very close to $1$) while maintaining customer fairness (small $Y$ and $std_\phi$) in both cases.

\section{Improving FairRec Through Envy-cycle Elimination}\label{sec:modified_fairrec}
While FairRec provides two-sided fair recommendations, it can be further tweaked to improve the recommendation performance 
for the customers. We propose such a modification of FairRec in \cref{subsec:modified_algo} as FairRecPlus, and then evaluate it against FairRec in~\cref{subsec:modified_algo_results}.
However, the improvement in FairRecPlus comes at a cost of an increased computation time. Thus, in scenarios where platforms can afford more time to compute a recommendation, FairRecPlus can be used to improve customer-side performance while still maintaining the same fairness guarantees on both sides.
\subsection{FairRecPlus}
\label{subsec:modified_algo}

	{
	\begin{algorithm}[!h]
		{\raggedright
			{
				{\raggedright{\bf Input:} Set of customers $U=[m]$, set of distinct products $P=[n]$,   recommendation set size $k$ (such that $k<m$ and $n\leq k\cdot m$), and the relevance scores $V_u(p)$.}\\
				{\raggedright{\bf Output:} A two-sided fair recommendation.}
			}
			\caption{\textit{FairRecPlus} ($U, P, k, V$)}
			\label{alg:mod-two-sided}
			\begin{algorithmic}[1]
				
				\State Initialize allocation $\mathcal{A}^0=(A^0_1, \ldots, A^0_m)$ with $A^0_i~\leftarrow~\emptyset$ for each customer $i \in [m]$.
				
				\State
				\State {\textbf{First Phase:}}
				\State Fix an (arbitrary) ordering of the customers $\sigma = \left(\sigma(1), \sigma(2),\ldots, \sigma(m)\right)$.
				\State Initialize set of feasible products $F_u \leftarrow P$ for each $u\in U$.
				\State Set $\ell \leftarrow \left\lfloor \frac{\alpha mk}{n}\right\rfloor$ denoting number of copies of each product.	
				\State Initialize each component of the vector $S=(S_1, \ldots, S_n)$ with $S_j\leftarrow \ell$, $\forall j\in [n]$, this stores the number of available copies of each product.
				\State Set $T\leftarrow \ell\times n$, total number of items to be allocated.
				\State $[\mathcal{B}, F]\leftarrow$Modified-Greedy-Round-Robin$(m,n,S,T,V,\sigma,F)$.
				\State Assign $\mathcal{A}\leftarrow \mathcal{A}\cup \mathcal{B}$.
				\State
				\State {\textbf{Second Phase:}}
				\For{$i=1$ to $m$}	
					\If{ $|A_i|<k$}
						\State Set $H$ as the top ($k-|A_i|$) items from $F_{i}$ (based on $V_i(\cdot)$ scores).
						\State Update $A_{i} \leftarrow A_{i} \cup H$. 
					\EndIf
				\EndFor
						
%
				
				\State Return $\mathcal{A}$.
			\end{algorithmic}}
		\end{algorithm}
	}
	
		{
		\begin{algorithm}[!h]
			{\raggedright
				{
					{\raggedright {\bf Input :} Number of customers $m$, number of producers $n$, an array with number of available copies of each product $S$, total number of available products $T>0$, relevance scores $V_u(p)$ and feasible product set $F_u$ for each customer, and an ordering $\sigma$ of $[m]$.}\\
					{\raggedright {\bf Output:} An allocation of $T$ products among $m$ customers and the residual feasible set $F_u$.
					}
					\caption{Modified-Greedy-Round-Robin ($m,n,S,T,V,\sigma,F$)}
					\label{alg:mod-greedy1}
					\begin{algorithmic}[1]
						\State Initialize allocation $\mathcal{B}=(B_1, \ldots, B_m)$ with $B_i~\leftarrow~\emptyset$ for each customer $i \in [m]$.	
						\State Initiate round $r\leftarrow 0$.
						\While{ true }
							\State Set $r\leftarrow r+1.$
							\For{$i =1 \mbox{ to } m$}
							\State Set $p \in \underset{p'\in F_{\sigma(i)}:(S_p'\neq 0)}{\argmax} V_{\sigma(i)}(p')$ 
							\If{$p == \emptyset$}
								\State Update $\mathcal{B} = \{B_1,\ldots,B_m\}$ to obtain an acyclic envy graph $\mathcal{G}(\mathcal{B})$ using Lemma~\ref{lemma:envylipton}.
								\State Update ordering $\sigma$ to be the topological ordering of the envy graph $\mathcal{G}(\mathcal{B})$.
								\State \textbf{go to} Step $24$.
							\EndIf
							\State Update $B_{\sigma(i)} \leftarrow B_{\sigma(i)} \cup p$. 
							\State Update $F_{\sigma(i)} \leftarrow F_{\sigma(i)} \setminus p$.
							\State Update $S_p\leftarrow S_p-1$.
							\State Update $T\leftarrow T-1$.					
							\If{ $T==0$ }
								\State Update $\mathcal{B} = \{B_1,\ldots,B_m\}$ to obtain an acyclic envy graph $\mathcal{G}(\mathcal{B})$ using Lemma~\ref{lemma:envylipton}.
								\State Update ordering $\sigma$ to be the topological ordering of the envy graph $\mathcal{G}(\mathcal{B})$.
								\State \textbf{go to} Step $24$.
							\EndIf
							\EndFor
										
							\State Update $\mathcal{B} = \{B_1,\ldots,B_m\}$ to obtain an acyclic envy graph $\mathcal{G}(\mathcal{B})$ using Lemma~\ref{lemma:envylipton}.
						\EndWhile			
						\State Return $\mathcal{B}=(B_1,\ldots,B_m)$ and $F=(F_1,\ldots,F_m)$.
					\end{algorithmic}}}
				\end{algorithm}
			}
			
\begin{figure*}[t!]
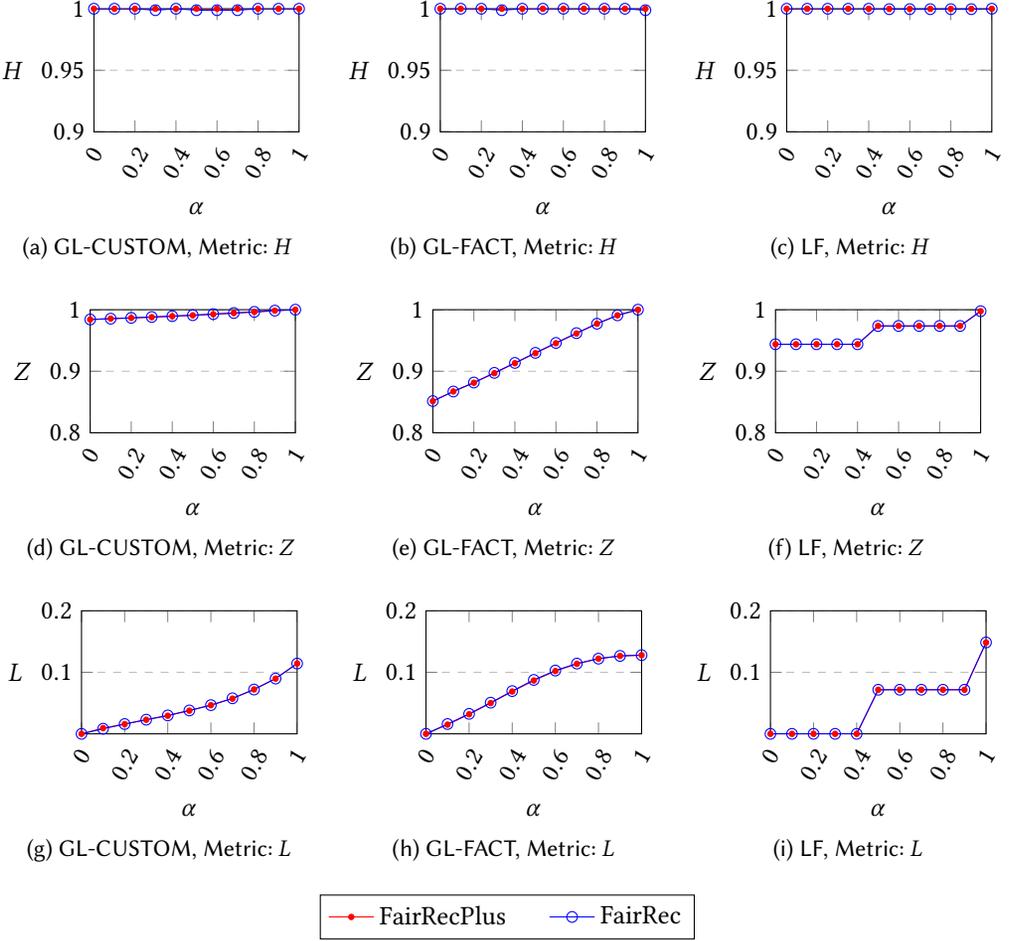

	\center{
		\subfloat[{GL-CUSTOM, Metric: \hyperref[metric_H]{$H$}}]{\pgfplotsset{width=0.31\textwidth,height=0.23\textwidth,compat=1.9}
			\input{figures/gl_1_H_mod}\label{fig:gl_1_H_mod}}
		\hfil
		\subfloat[{GL-FACT, Metric: \hyperref[metric_H]{$H$}}]{\pgfplotsset{width=0.31\textwidth,height=0.23\textwidth,compat=1.9}
			\input{figures/gl_2_H_mod}\label{fig:gl_2_H_mod}}
		\hfil
		\subfloat[{LF, Metric: \hyperref[metric_H]{$H$}}]{\pgfplotsset{width=0.31\textwidth,height=0.23\textwidth,compat=1.9}
			\input{figures/lf_H_mod}\label{fig:lf_H_mod}}
		\hfil
		\subfloat[{GL-CUSTOM, Metric: \hyperref[metric_Z]{$Z$}}]{\pgfplotsset{width=0.31\textwidth,height=0.23\textwidth,compat=1.9}
			\input{figures/gl_1_Z_mod}\label{fig:gl_1_Z_mod}}
		\hfil
		\subfloat[{GL-FACT, Metric: \hyperref[metric_Z]{$Z$}}]{\pgfplotsset{width=0.31\textwidth,height=0.23\textwidth,compat=1.9}
			\input{figures/gl_2_Z_mod}\label{fig:gl_2_Z_mod}}
		\hfil
		\subfloat[{LF, Metric: \hyperref[metric_Z]{$Z$}}]{\pgfplotsset{width=0.31\textwidth,height=0.23\textwidth,compat=1.9}
			\input{figures/lf_Z_mod}\label{fig:lf_Z_mod}}
		\hfil
		\subfloat[{GL-CUSTOM, Metric: \hyperref[metric_L]{$L$}}]{\pgfplotsset{width=0.32\textwidth,height=0.23\textwidth,compat=1.9}
			\input{figures/gl_1_L_mod}\label{fig:gl_1_L_mod}}
		\hfil
		\subfloat[{GL-FACT, Metric: \hyperref[metric_L]{$L$}}]{\pgfplotsset{width=0.32\textwidth,height=0.23\textwidth,compat=1.9}
			\input{figures/gl_2_L_mod}\label{fig:gl_2_L_mod}}
		\hfil
		\subfloat[{LF, Metric: \hyperref[metric_L]{$L$}}]{\pgfplotsset{width=0.32\textwidth,height=0.23\textwidth,compat=1.9}
			\input{figures/lf_L_mod}\label{fig:lf_L_mod}}
		\vfil
		\subfloat{\pgfplotsset{width=.5\textwidth,compat=1.9}
			\begin{tikzpicture}
			\begin{customlegend}[legend entries={{FairRecPlus},{FairRec}},legend columns=5,legend style={/tikz/every even column/.append style={column sep=0.5cm}}]			
			\addlegendimage{red,mark=*,mark size=1pt,sharp plot}
			\addlegendimage{blue,mark=o,sharp plot}
			\end{customlegend}
			\end{tikzpicture}}
	}\caption{FairRec modification results on producer-side for $k=20$. First row: fraction of satisfied producers (\hyperref[metric_H]{$H$}). Second row: inequality in producer exposures (\hyperref[metric_Z]{$Z$}). Third row: exposure loss on producers (\hyperref[metric_L]{$L$}).}\label{fig:modified_fairrec_prod}
\end{figure*}

\begin{figure*}[t!]
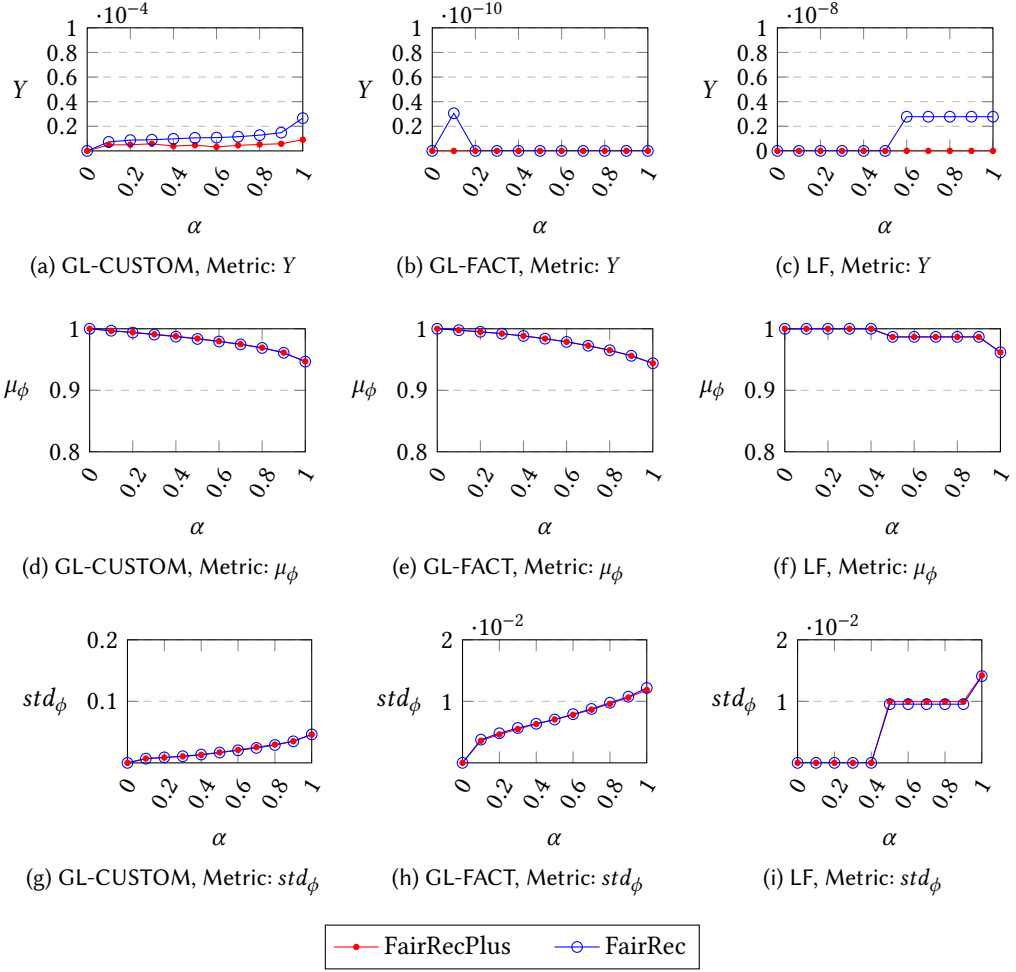

	\center{
		\subfloat[{GL-CUSTOM, Metric: \hyperref[metric_Y]{$Y$}}]{\pgfplotsset{width=0.32\textwidth,height=0.23\textwidth,compat=1.9}
			\input{figures/gl_1_Y_mod}\label{fig:gl_1_Y_mod}}
		\hfil
		\subfloat[{GL-FACT, Metric: \hyperref[metric_Y]{$Y$}}]{\pgfplotsset{width=0.32\textwidth,height=0.23\textwidth,compat=1.9}
			\input{figures/gl_2_Y_mod}\label{fig:gl_2_Y_mod}}
		\hfil
		\subfloat[{LF, Metric: \hyperref[metric_Y]{$Y$}}]{\pgfplotsset{width=0.32\textwidth,height=0.23\textwidth,compat=1.9}
			\input{figures/lf_Y_mod}\label{fig:lf_Y_mod}}
		\hfil
		\subfloat[{GL-CUSTOM, Metric: \hyperref[sec:cust_metric]{${\mu_{\phi}}$}}]{\pgfplotsset{width=0.32\textwidth,height=0.23\textwidth,compat=1.9}
			\input{figures/gl_1_mu_mod}\label{fig:gl_1_mu_mod}}
		\hfil
		\subfloat[{GL-FACT, Metric: \hyperref[sec:cust_metric]{${\mu_{\phi}}$}}]{\pgfplotsset{width=0.32\textwidth,height=0.23\textwidth,compat=1.9}
			\input{figures/gl_2_mu_mod}\label{fig:gl_2_mu_mod}}
		\hfil
		\subfloat[{LF, Metric: \hyperref[sec:cust_metric]{${\mu_{\phi}}$}}]{\pgfplotsset{width=0.32\textwidth,height=0.23\textwidth,compat=1.9}
			\input{figures/lf_mu_mod}\label{fig:lf_mu_mod}}
		\hfil
		\subfloat[{GL-CUSTOM, Metric: \hyperref[sec:cust_metric]{${std_\phi}$}}]{\pgfplotsset{width=0.29\textwidth,height=0.23\textwidth,compat=1.9}
			\input{figures/gl_1_sigma_mod}\label{fig:gl_1_sigma_mod}}
		\hfil
		\subfloat[{GL-FACT, Metric: \hyperref[sec:cust_metric]{${std_\phi}$}}]{\pgfplotsset{width=0.29\textwidth,height=0.23\textwidth,compat=1.9}
			\input{figures/gl_2_sigma_mod}\label{fig:gl_2_sigma_mod}}
		\hfil
		\subfloat[{LF, Metric: \hyperref[sec:cust_metric]{${std_\phi}$}}]{\pgfplotsset{width=0.29\textwidth,height=0.23\textwidth,compat=1.9}
			\input{figures/lf_sigma_mod}\label{fig:lf_sigma_mod}}
		\vfil
		\subfloat{\pgfplotsset{width=.5\textwidth,compat=1.9}
			\begin{tikzpicture}
			\begin{customlegend}[legend entries={{FairRecPlus},{FairRec}},legend columns=5,legend style={/tikz/every even column/.append style={column sep=0.5cm}}]
			\addlegendimage{red,mark=*,mark size=1pt,sharp plot}
			\addlegendimage{blue,mark=o,sharp plot}
			\end{customlegend}
			\end{tikzpicture}}
	}\caption{FairRec modification results on customer-side for $k=20$. First row: mean average envy (\hyperref[metric_Y]{$Y$}). Second row: mean customer utility (\hyperref[sec:cust_metric]{$\mu_\phi$}). Third row: standard deviation of customer utilities (\hyperref[sec:cust_metric]{$std_\phi$}).}\label{fig:modified_fairrec_cust}
\end{figure*}

\noindent
We now present a modification of FairRec: named {\it FairRecPlus}. FairRecPlus executes in two phases, similar to FairRec. The {\bf first phase} creates $\ell=\left\lfloor\frac{mk}{n}\right\rfloor$ copies of each product and then initializes each component of the vector $S$ of size $|P|$ to the value $\ell$. Then, assuming an arbitrary ordering $\sigma$ of customers, the modified greedy algorithm $\ALG~\ref{alg:mod-greedy1}$ is executed, which is  different from our earlier approach. In this modification, at the end of a round of greedy-round-robin allocation, we create an \textit{envy graph} and compute the topological ordering among the agents based on their partial allocation. The high-level idea is to re-order the priorities of the agents during round-robin allocations, aiming to balance the extent of envy between each pair of agents by maintaining an acyclic envy-graph after each round. An envy graph is a directed graph that captures the envy between agents---the nodes in the envy graph represent the agents and it contains a directed edge from $i$ to $j$ if and only if, $i$ envies $j$, i.e., if and only if $v_i(B_i) < v_i(B_j)$, where $B_i$ and $B_j$ are partial allocations to agents $i$ and $j$, respectively. It was established in \cite{lipton-envy-graph} that one can always efficiently update a given partial allocation such that the resulting envy graph is acyclic.

\begin{lemma}\label{lemma:envylipton}
	(\citet{lipton-envy-graph}) Given a partial allocation $(A_1, \ldots, A_m)$, we can find another partial allocation $\mathcal{B}$=$(B_1, \ldots, B_m)$ in polynomial time such that \\
	(i) The valuations of the agents for their bundles do not decrease: $v_i(B_i) \geq v_i(A_i)$ for all $i \in [m]$.\\
	(ii) The envy graph $G(\mathcal{B})$ is acyclic.
\end{lemma}
\textit{Proof Sketch.}
If the envy graph of $\mathcal{A}$ is acyclic then the claim holds trivially. Otherwise, find a cycle in the graph $G(\mathcal{A})$ (time complexity $\mathcal{O}(n+|E|)$, where $E$ is the number of edges in the graph). Let $C=i_1 \rightarrow i_2 \rightarrow \ldots \rightarrow i_k \rightarrow i_1$ be the cycle. The bundles can be reallocated as follows: for all agents not in $C$, i.e., $k \notin \{i_1, i_2, \ldots, i_k \}$ set $B_k = A_k$, and for all the agents in the cycle set $B_i$ to be the bundle of their successor in $C$, i.e., set $B_{i_a} = A_{i_{(a+1)}}$ for $ 1\leq a < k$ along with $B_{i_k} = A_{i_1}$. After this reallocation $v_i(B_i) \geq v_i(A_i)$ for all $i \in [n]$. Furthermore, the number of edges in $G(\mathcal{B})$ is strictly less than $G(\mathcal{A})$: the directed edges $i_1 \rightarrow i_2 \rightarrow \ldots \rightarrow i_k \rightarrow i_1$ do not appear in the envy graph of $(B_1, \ldots, B_n)$ and if an agent $k$ starts envying an agent in the cycle, say agent $i_a$, then $k$ must have been envious of $i_{a+1}$ in $\mathcal{A}$. Edges between agents $k$ and $k'$ which are not in the cycle $C$ remain unchanged, and edges going out of an agent $i$ in the cycle $C$ can only get removed, since $i$'s valuation for the bundle assigned to her bundle increases. Therefore, we can repeatedly remove cycles and keep reducing the number of edges in the envy graph to eventually find a partial allocation $\mathcal{B}$ that satisfies the stated claim. $\hfill\square$\\

The worst-case time complexity of eliminating envy cycles, to obtain a directed acyclic graph (DAG) at each round, is $O(m^4)$. A topological ordering of the acyclic directed graph $\mathcal{G}(\mathcal{B})$ can be computed for updating the $\sigma$ in $O(m)$ time. This new ordering is then used for allocating the next round of items in a round-robin manner. 

The {\bf second phase} checks if all the customers have received exactly $k$ products. If yes, then no further allocation is required; if not, then allocate each agent $i$, their most valuable items from $F_i$ until they receive $k$ items. 

The time complexity of FairRecPlus is governed by the envy cycle elimination step which takes $O(m^4k)$ over $k$ rounds. Moreover, for each of the $mk$ items, finding the maximum valued feasible producer takes $O(n)$ time. Thus, the total time complexity of FairRecPlus is $O(m^4k + mnk)$. Therefore, for a large number of customers, FairRecPlus would take a huge time to compute a two-sided fair allocation. 
Although this modification requires more computation, we empirically observe that, it improves on the customer-side metrics.

\subsection{Results with FairRecPlus}
\label{subsec:modified_algo_results}
We test the FairRecPlus algorithm on all the datasets with $k=20$ while varying $\alpha$ from $0$ to $1$ in separate trials. The results on producer-side and customer-side are plotted in \cref{fig:modified_fairrec_prod,fig:modified_fairrec_cust} respectively.
While the producer-side plots for FairRec and FairRecPlus (\cref{fig:modified_fairrec_prod}) seem to completely overlap, they are marginally different from each other; for example, at $\alpha=0.5$ in GL-CUSTOM, the $H$, $Z$, and $L$ metrics for FairRecPlus are $1$, $0.9908$, and $0.0376$ respectively while those in FairRec are $0.99$, $0.9910$, and $0.0380$ respectively.
The producer-side performances of FairRecPlus is very similar to that of FairRec since the modification does not change any guarantee on the producer side.
On the other hand, the modification introduces only the envy-cycle removal rounds which reduces the envy on customer-side.
We observe that FairRecPlus reduces the mean customer envy (check $Y$ in \cref{fig:gl_1_Y_mod,fig:gl_2_Y_mod,fig:lf_Y_mod}). 
On the other hand, the performances of FairRecPlus in other customer-side metrics are similar to those of FairRec (\cref{fig:gl_1_mu_mod,fig:gl_2_mu_mod,fig:lf_mu_mod,fig:gl_1_sigma_mod,fig:gl_2_sigma_mod,fig:lf_sigma_mod});
for example at $\alpha=0.5$ in GL-CUSTOM, the mean and standard deviation of customer utilities are $0.9841$ and $0.0169$ in FairRecPlus, against the corresponding values $0.9834$ and $0.0167$ in FairRec.
\section{conclusion}\label{discussion}
In this work, we propose the notion of two-sided fairness for recommendations in two-sided platforms.
For producers, we consider a minimum exposure guarantee while we try to ensure less inequality in customer utilities.
Note that we assume the relevance of a product does not play any role in producer's utility (in contrast to \citet{singh2018fairness,biega2018equity}), and use only the exposure of a producer as her utility.
We provide a scalable and easily adaptable algorithm that exhibits desired two-sided fairness properties while causing a marginal loss in the overall quality of recommendations. We establish theoretical guarantees and provide empirical evidence through extensive evaluations of real-world datasets. Furthermore, we propose a modification of our algorithm and show that it performs better on customer-side metrics while being two-sided fair, but at the cost of additional computation time. 
Our work can be directly applied to fair recommendation problems in scenarios like mass recommendation/promotion sent through emails, app/web notifications. Though our work considers the offline recommendation scenario where the recommendations are computed for all the registered customers at once, it can also be extended for online recommendation settings by limiting the set of customers to only the active customers at any particular instant. However, developing a more robust realization of the proposed mechanism for a completely online scenario remains future work. Going ahead, we also want to study attention models that can handle position bias~\cite{agarwal2019estimating}, where customers pay more attention to the top-ranked products than the lower-ranked ones.\\

\noindent {\bf Acknowledgements:}
This work was conducted when A. Biswas was a PhD student at the Indian Institute of Science. She gratefully acknowledges the support of a Google PhD Fellowship Award.
G. K Patro acknowledges the support by TCS Research Fellowship.
This research was supported in part by an European Research Council (ERC) Advanced Grant for the project ``Foundations for Fair Social Computing" (grant agreement no. 789373), and an European Research Council (ERC) Marie Sklodowska-Curie grant for the project ``NoBIAS --- Artificial Intelligence without Bias" (grant agreement no. 860630), both funded under the EU's Horizon 2020.\\

\noindent {\bf Reproducibility:}
Code, dataset and other details are available at \url{https://github.com/gourabkumarpatro/FairRec}.


\bibliographystyle{ACM-Reference-Format}
\bibliography{main}

\end{document}